\documentclass[a4paper,reqno]{amsart}
\pdfoutput=1
\usepackage[pdftex]{thumbpdf,hyperref}
\newtheoremstyle{normal}{}{}{}{}{\bfseries}{}{.5em}{\thmname{#1}\thmnumber{ #2}\thmnote{ (#3)}.}
\theoremstyle{normal}
\newtheorem{lemm}{Lemma}

\newtheorem{rem}{Remark}

\newtheorem{example}{Example}
\makeatletter
\@addtoreset{equation}{section}
\makeatother

\usepackage[applemac]{inputenc}
\usepackage{url}
\usepackage{amsmath,amsfonts,amssymb,mathrsfs,dsfont,amsthm}
\usepackage{graphicx}
\usepackage{longtable}
\usepackage{subfigure}
\usepackage{enumerate}
\usepackage{color}
\usepackage[all,knot,poly]{xy}
\usepackage{slashed}
\usepackage{youngtab}
\usepackage[colorinlistoftodos]{todonotes}
\newcommand{\Sym}{\mathcal S}
\newcommand{\Lin}{\mathcal{L}}

\newcommand{\Knot}{\mathcal K}
\renewcommand{\bar}[1]{\overline{#1}}
\renewcommand{\hat}[1]{\widehat{#1}}

\newcommand{\intint}[1]{\left[\kern-0.15em\left[#1\right]\kern-0.15em\right]}
\newcommand{\into}{\ensuremath{\lhook\joinrel\relbar\joinrel\rightarrow}}
\newcommand{\R}{\mathbb{R}}

\newcommand{\N}{\mathbb{N}}
\newcommand{\ZZ}{\mathbb{Z}}

\renewcommand{\vec}[1]{\mathbf{#1}}
\newcommand{\norm}[1]{\Vert#1 \Vert}
\newcommand{\abs}[1]{\vert#1 \vert}

\newcommand{\bigsqplus}{\qed\kern-0.75em+} 
\newcommand{\sqplus}{\qed\kern-1em+}
\newcommand{\semsum}{\odot\kern-0.75em\ltimes}

\newcommand{\tens}[1]{\mathsf{\mathbf{#1}}}

\DeclareMathOperator{\Tr}{\mathrm{Tr}}

\def\XXint#1#2#3{{\setbox0=\hbox{$#1{#2#3}{\int}$}\vcenter{\hbox{$#2#3$}}\kern-.5\wd0}}
\newcommand{\UD}[1]{\mathcal{D}#1}
\newcommand{\ud}[1]{\mathrm{d}#1}

\newcommand{\bra}[1]{\langle #1|}
\newcommand{\ket}[1]{|#1\rangle}
\newcommand{\av}[1]{\langle #1\rangle}

\newcommand{\Part}{\mathscr{P}}
\newcommand{\T}{\mathcal{T}}
\newcommand{\PO}{\mathcal{P}}

\renewcommand{\S}{\mathbb{S}}
\newcommand{\Torus}{\mathbb{T}}

\newcommand{\so}{\mathfrak{so}}

\newcommand{\g}{\mathfrak{g}}

\newcommand{\W}{\mathcal{W}}
\DeclareMathOperator{\ch}{\mathrm{ch}}

\newcommand{\Hilb}{\mathcal{H}}
\newcommand{\AmSLaTeX}{{\protect\the\textfont2 A}%
\kern-.1667em\lower.5ex\hbox {\protect\the\textfont2 M}%
\kern-.125em{\protect\the\textfont2 S}-\LaTeX}
\newcommand{\figref}[1]{\textsc{Fig}.~\ref{#1}}
\newcommand{\tabref}[1]{\textsc{Tab}.~$\ref{#1}$}

\title{Chern--Simons Invariants of Torus Links}
\author{S\'ebastien Stevan}
\address{Section de math\'ematiques\\
Universit\'e de Genève\\
Case postale 64\\
1211 Genève 4 (Switzerland)}
\email{Sebastien.Stevan@unige.ch}
\begin{document}
\begin{abstract}
We compute the vacuum expectation values of torus knot operators in Chern--Simons theory, and we obtain explicit formulae for all classical gauge groups and for arbitrary representations. We reproduce a known formula for the HOMFLY invariants of torus knots and links, and we obtain an analogous formula for Kauffman invariants. We also derive a formula for cable knots. We use our results to test a recently proposed conjecture that relates HOMFLY and Kauffman invariants.
\end{abstract}
\maketitle
\section{Introduction}
The idea of using Chern--Simons theory \cite{Chern74a} to compute knot invariants goes back to Witten's paper \cite{Witten89a} in 1989, when he identified the skein relation satisfied by the Jones polynomial \cite{Jones85a}. Though the theory is in principle exactly solvable, the computations are quite challenging in most cases. One convenient framework to address such problems is the formalism of knot operators \cite{Labastida89a}. For torus knots, an explicit operator formalism has been constructed by \cite{Labastida91a}, that successfully reproduces the Jones polynomial for Wilson loops carrying the fundamental representation of $SU(2)$.

Several further works have generalized the computation to arbitrary representations of $SU(2)$ \cite{Isidro93a}, to the fundamental representation of $U(N)$ \cite{Labastida95a} and to arbitrary representations of $U(N)$ \cite{Labastida01a}. There have also been attempts to compute Kauffman invariants from Chern--Simons theory. With Wilson loops carrying the fundamental representation of $SO(N)$, Labastida and Pérez  obtained a simple formula for the Kauffman polynomial \cite{Labastida96a}. For torus knots of the form $(2,2m+1)$, there are formulae for arbitrary representations of $SO(N)$ \cite{Devi93a,Borhade05a}, but they are not completely explicit due to the presence of a generally unknown group-theoretic sign.

Recently, a simple formula for HOMFLY invariants of torus links has been obtained by using quantum groups methods \cite{Lin06a}. For quantum Kauffman invariants, L. Chen and Q. Chen \cite{Chen09a} had derived a similar formula but published it only after this paper was submitted. These results encouraged us to address the computation of torus link invariants from Chern--Simons point of view. In this paper, we carefully analyze the matrix elements of knot operators to produce simpler formulae. Our approach uses only group-theoretic data and is valid for any gauge group. As an application, we compute the polynomial invariants for all classical Lie groups and for arbitrary representations, and we reproduce the results of \cite{Lin06a}.

As explicit formulae are available, torus knots represent an useful ground to test the conjectured relationship between knot invariants and string theory. The equivalence of  $1/N$ expansion of Chern--Simons theory to topological string theory \cite{Gopakumar99a} implies that the colored HOMFLY polynomial can be related to Gromov--Witten invariants, and thus enjoys highly nontrivial properties \cite{Ooguri00a,Labastida00a}. This conjecture has been extensively checked \cite{Labastida00a,Labastida01a,Lin06a}, and is now proved \cite{Liu07a}. The large-$N$ duality of Chern--Simons theory with gauge group $SO(N)$ or $Sp(N)$ has also been studied \cite{Sinha00a}. In \cite{Bouchard05a}, partial conjectures on the structure of Kauffman invariants have been formulated. The complete conjecture, that also involves HOMFLY invariants for composite representations, has been stated by Mariño \cite{Marino09c}.

The outline of the paper is as follows: in Section 2, we recall some important properties of Wilson loops. Section 3 is devoted to the matrix elements of torus knot operators. In Sections 4,5 and 6, we deduce explicit formulae for HOMFLY and Kauffman invariants of cable knots, torus knots and torus links. Finally, in Section 7 we provide some tests of Mariño's conjecture.

\section{Chern--Simons Theory and Wilson Loop Operators}\label{cs}
Chern-Simon theory is a topological gauge theory on an orientable, boundaryless $3$-manifold $M$ with a simple, simply connected, compact, nonabelian Lie group $G$ and the action
\begin{equation}
\label{cs_action}
S(\vec A)=\frac{k}{4\pi}\int_M\Tr\Big[ \vec A\wedge\ud\vec A+\frac 2 3 \vec A\wedge\vec A\wedge \vec A\Big]
\end{equation}
where $\Tr$ is the trace in the fundamental representation and $k$ is a real parameter. In this expression $\vec A$ is a $\g$-valued $1$-form on $M$, where $\g$ is the Lie algebra of the gauge group $G$.

In the context of knot invariants, $M$ is usually taken to be $\S^3$ and the relevant gauge-invariant observables are Wilson loop operators. Let $\Knot\subset\S^3$ be a knot and $V_\lambda$ an irreducible $\g$-module of highest weight $\lambda$. The associated Wilson loop is 
\begin{equation}
\label{wilson}
\tens W_{\lambda}^\Knot(\vec A)=\Tr_{V_\lambda}\Big[\PO\!\exp\oint_\Knot\vec A\Big],
\end{equation}
where $\PO\!\exp$ is a path-ordered exponential. In other words $\tens W^\lambda_\Knot(\vec A)$ is obtained by taking the trace on $V_\lambda$ of the holonomy along $\Knot$.

As was realized first by Witten \cite{Witten89a}, the vacuum expectation value (vev)
\begin{equation}
\label{vev}
\av{\tens W_{\lambda_1}^{\Knot_1}\cdots \tens W_{\lambda_L}^{\Knot_L}}=\frac{\int\UD[\vec A]\,\tens W_{\lambda_1}^{\Knot_1}(\vec A)\cdots \tens W_{\lambda_L}^{\Knot_L}(\vec A)e^{iS(\vec A)}}{\int\UD[\vec A]\,e^{iS(\vec A)}},
\end{equation}
where the functional integration runs over the gauge orbits of the field, is a framing-dependent invariant of the link $\Lin=\Knot_1\cup\dots\cup\Knot_L$.

Indeed $W_\lambda(\Knot)=\av{\tens W_\lambda^\Knot}$ reproduces the quantum invariant obtained from the category of $U_q(\g)$-modules. In this paper we shall encounter colored HOMFLY invariants $H_\lambda^{\Knot}(t,v)$ corresponding to the group $U(N)$ and colored Kauffman invariants $K_\lambda^\Knot(t,v)$ corresponding to the groups $SO(N)$ and $Sp(N)$.

The vev \eqref{vev} can be computed perturbatively or by nonperturbative methods based on surgery of $3$-manifolds. In this paper we consider these later methods, in particular the formalism of knot operators. Before turning to knot operators, and restricting to torus knots, we review some properties of Wilson loops.

\subsection{Product of Wilson loops with the same orientation}
We provisorily take $G$ to be $U(N)$ for definiteness. Representations that label Wilson loops are usually polynomial representations (those indexed by partitions). When we write $\tens W_\lambda^\Knot$ for a Wilson loop or $W_\lambda(\Knot)$ for an invariant, we implicitly assume that the representation with highest weight $\lambda\in\Lambda_W^+$ is polynomial, so that we can symbolize $\lambda$ by a partition.

The first relation to be mentioned is the well-known fusion rule for Wilson loops. For an oriented link made of two copies of the same knot, with the same orientation for both components (as in \figref{torus_link_1} for instance), one has
\begin{equation}
\label{fusion}
\av{\tens W_\lambda^\Knot\tens W_\mu^\Knot}=\sum_{\nu\in\Part}N_{\lambda\mu}^\nu\av{\tens W_\nu^\Knot},
\end{equation}
where $\Part$ is the set of nonempty partitions and $N_{\lambda\mu}^\nu$ are the coefficients in the decomposition of the tensor product
\begin{equation*}
V_\lambda\otimes V_\mu=\bigoplus_{\nu\in\Part}N_{\lambda\mu}^\nu V_\nu.
\end{equation*}
They are called Littlewood--Richardson coefficients for $U(N)$.

Formula \eqref{fusion} is extremely useful, since it reduces any product of Wilson loops that share the same orientation to a sum of Wilson loops. It only applies to links composed by several copies of the same knot, but this is not a restriction for torus links.

For other Lie groups the same formula holds with different coefficients. For $SO(N)$ and $Sp(N)$ they are given by \cite{Littlewood40a,Koike89a}
\begin{equation}
\label{LRSO}
M_{\lambda\mu}^{\nu}=\sum_{\alpha,\beta,\gamma} N_{\alpha\beta}^\lambda N_{\alpha\gamma}^\mu N_{\beta\gamma}^\nu.
\end{equation}
Here the sum runs over $\Part\cup\{\emptyset\}$.
\begin{rem}
Formula \eqref{fusion} has to be understood as a regularization for the product of two operators evaluated at the same point. It extends the relation
\begin{equation}
\label{fusion2}
\tens W_\lambda^\Knot(\vec A)\tens W_\mu^\Knot(\vec A)=\sum_{\nu\in\Part}N_{\lambda\mu}^\nu\tens W_\nu^\Knot(\vec A)
\end{equation}
between the functionals $\tens W_\lambda^\Knot(\vec A)$ to the quantized Wilson loops. We derive \eqref{fusion2} by noting that the holonomy $\tens U_\Knot$ is an element of $G$, hence it is conjugate to an element of the maximal torus of $G$ \cite{Knapp05a}. Furthermore $\Tr_{V_\lambda}$ is the character of $V_\lambda$ as a function of the eigenvalues, and the product of characters is decomposed as the tensor product of representation.
\end{rem}

\subsection{Product of Wilson loops with different orientations}
The need to consider all rational representations appears when one deals with both orientations for $\Knot$ (as in \figref{torus_link_2} for example). The product of two Wilson loops $\tens W^\Knot_\lambda$ and $\tens W^{-\Knot}_\mu$, where $-\Knot$ denotes $\Knot$ with the opposite orientation, cannot be decomposed as above. In the formalism of the HOMFLY skein of the annulus \cite{Hadji06a}, one would have to use the basis of the full skein, indexed by two partitions. In Chern--Simons theory the same role is played by composite representations.

\begin{figure}[htb]
\begin{center}
\subfigure[$\tens W^\Knot_\lambda\tens W^\Knot_\mu$]{\label{torus_link_1}\includegraphics[width=.3\linewidth]{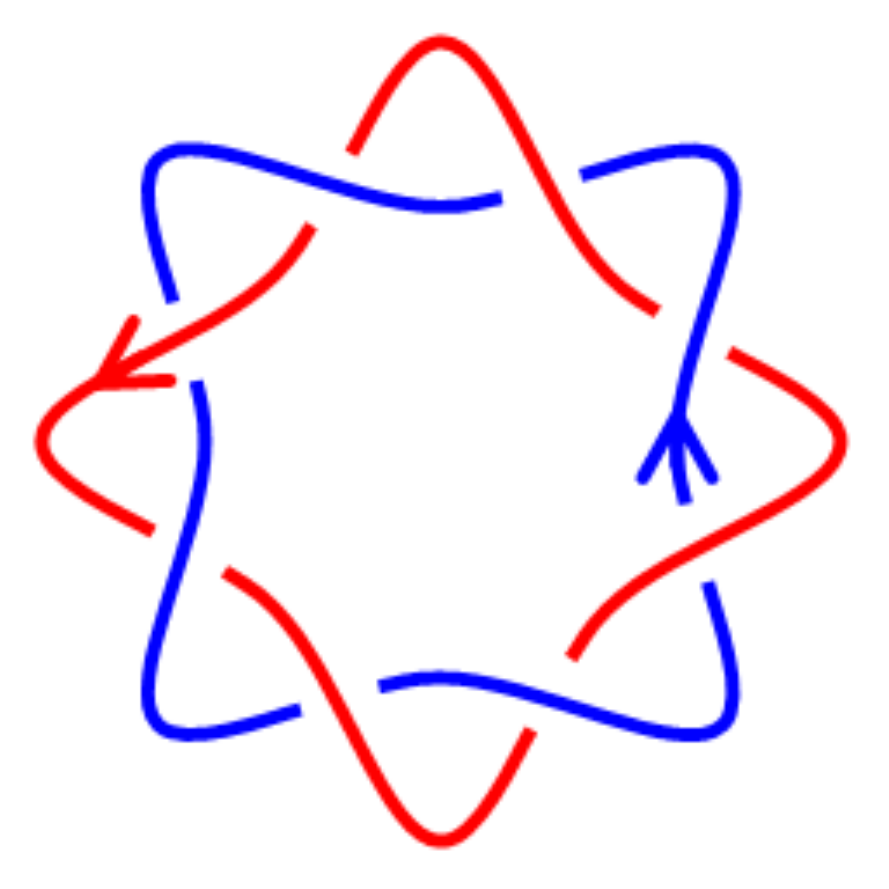}}\qquad\qquad\qquad
\subfigure[$\tens W^\Knot_\lambda\tens W^{-\Knot}_\mu$]{\label{torus_link_2}\includegraphics[width=.3\linewidth]{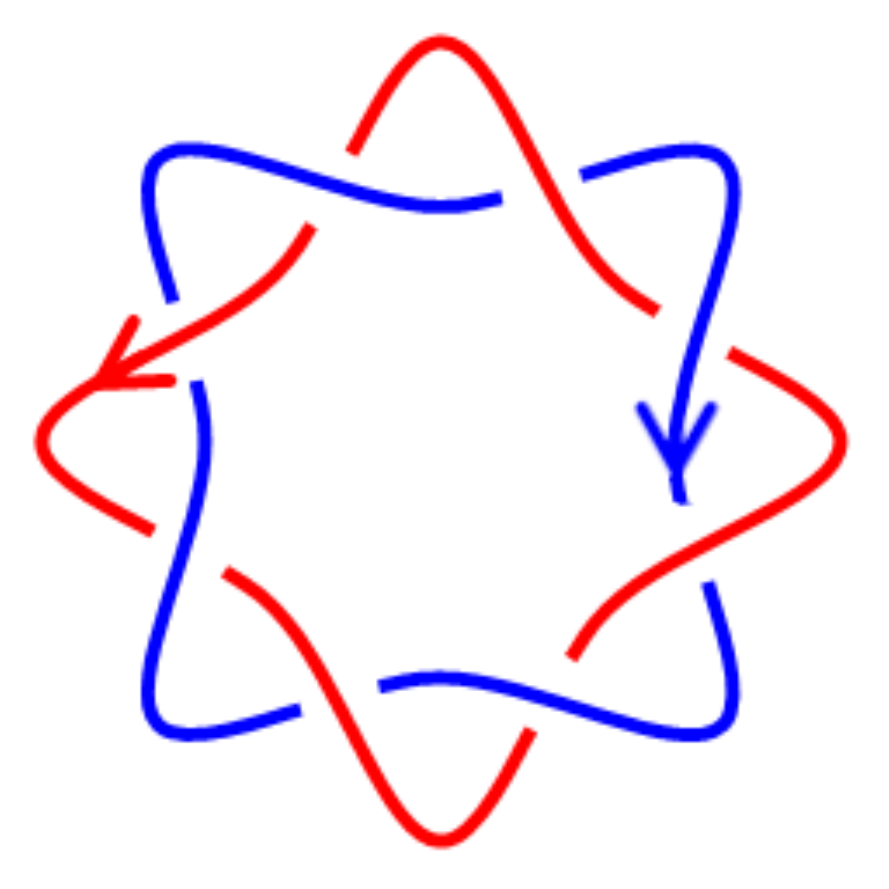}}
\caption{Products of Wilson loops with various orientations.}
\label{torus_link}
\end{center}
\end{figure}

Composite (or mixed tensor) representations
$$V_{[\lambda,\mu]}=\sum_{\eta,\nu,\zeta} (-1)^{\abs\eta} N_{\eta\nu}^\lambda N_{\bar\eta\zeta}^\mu V_\nu\otimes\bar{V_\zeta}$$
are the most general irreducible representations of $U(N)$, where the sum runs over partitions and $\bar\eta$ is the partition conjugate to $\eta$ (the transpose Young diagram). More details on composite representations can be found in \cite{Halverson96a}.

It is straightforward to derive a fusion rule for $\tens W^\Knot_\lambda\tens W^{-\Knot}_\mu$ by decomposing mixed tensor representations. Let $\tens U_\Knot$ be the holonomy along $\Knot$; then
\begin{align*}
\tens W^\Knot_\lambda\tens W^{-\Knot}_\mu&=\Tr_{V_\lambda}\tens U_\Knot\Tr_{V_\mu}\tens U_\Knot^{-1}\\
&=\Tr_{V_\lambda}\tens U_\Knot\Tr_{\bar{V_\mu}}\tens U_\Knot\\
&=\Tr_{V_\lambda\otimes\bar{V_\mu}}\tens U_\Knot.
\end{align*}
One has the following decomposition of  $V_\lambda\otimes\bar{V_\mu}$ in terms of composite representations \cite{Koike89a}
$$V_\lambda\otimes\bar{V_\mu}=\sum_{\eta,\nu,\zeta} N_{\eta\nu}^\lambda N_{\zeta\nu}^\mu V_{[\eta,\zeta]}.$$
If we denote by $\tens W_{[\eta,\zeta]}^\Knot$ the Wilson loop in the composite representation $V_{[\eta,\zeta]}$, we get the fusion rule
\begin{equation}
\label{fusion_composite}
\av{\tens W^\Knot_\lambda\tens W^{-\Knot}_\mu}=\sum_{\eta,\nu,\zeta} N_{\eta\nu}^\lambda N_{\zeta\nu}^\mu \av{\tens W_{[\eta,\zeta]}^\Knot}.
\end{equation}

\begin{rem}
Since $V_{[\lambda,\emptyset]}=V_\lambda$ and $V_{[\emptyset,\lambda]}=V_\lambda^*$, one has 
$$\tens W^\Knot_{[\lambda,\emptyset]}=\tens W^\Knot_\lambda\qquad\text{and}\qquad\tens W^{\Knot}_{[\emptyset,\lambda]}=\tens W^{-\Knot}_\lambda.$$
More generally $\tens W^\Knot_{[\lambda,\mu]}=\tens W^{-\Knot}_{[\mu,\lambda]}$.
\end{rem}

We can as well consider product of Wilson loops carrying composite representations and write a fusion rule for them. It is given by \cite{Koike89a}
$$\av{\tens W^\Knot_{[\lambda,\mu]}\tens W^{\Knot}_{[\eta,\nu]}}=\sum_{\alpha,\beta,\gamma,\delta}\sum_{\xi,\zeta} \Big(\sum_\kappa N_{\kappa\alpha}^\lambda N_{\kappa\beta}^\nu\Big)\Big(\sum_\epsilon N_{\epsilon\delta}^\mu N_{\epsilon\gamma}^\eta \Big) N_{\alpha\gamma}^\xi N_{\beta\delta}^\zeta\av{\tens W_{[\xi,\zeta]}^\Knot}.$$

\subsection{Traces of powers of the holomony}
As will be illustrated later in this paper, traces of powers of the holonomy along a given knot play an important in the gauge theory approach to knot invariants. In fact, such composite observables can be decomposed by a group-theoretic approach.

Given a knot $\Knot$, the holonomy $\tens U_\Knot$ is conjugate to an element in the maximal torus of $G$, and we already mentioned that
\begin{equation}
\label{eq2.1}
\Tr_{V_\lambda}\tens U_\Knot=\ch_\lambda(z_1,\dots,z_r),
\end{equation}
where $\ch_\lambda$ is the character of $\g$ and $z_1,\dots,z_r$ are the variable eigenvalues of $\tens U_\Knot$ ($r$ is the rank of $G$).

The trace of the $n$-th power of the holonomy is then given by
\begin{equation}
\label{eq2.2}
\Tr_\lambda\tens U^n_\Knot=\ch_\lambda(z^n_1,\dots,z^n_r).
\end{equation}
Let $\Lambda_W$ be the weight lattice and $\W$ the Weyl group of $G$. Equation \eqref{eq2.2} is obtained from \eqref{eq2.1} by applying the ring homomorphism
\begin{equation*}\begin{array}{rccc}
\Psi_n:& \ZZ[\Lambda_W]^\W & \longrightarrow & \ZZ[\Lambda_W]^\W\\
& e^\mu& \longmapsto & e^{n\mu}\end{array}
\end{equation*}
which is called the Adams operation. Since the characters form a $\ZZ$-basis of $\ZZ[\Lambda_W]^\W$, there exist integer coefficients  $c_{\lambda,n}^\nu$ univocally determined by the decomposition of $\Psi_n\ch_\lambda$ with respect to the basis $(\ch_\nu)_{\nu\in\Lambda_W^+}$:
\begin{equation}
\label{adams}
\Psi_n\ch_\lambda=\sum_{\nu\in\Lambda_W^+} c_{\lambda,n}^\nu\ch_\nu.
\end{equation}
Hence we have obtained the following formula:
\begin{equation}
\label{adams_wilson}
\Tr_\lambda\tens U_\Knot^n=\sum_{\nu\in\Part}c_{\lambda,n}^\nu\Tr_\nu\tens U_\Knot.
\end{equation}
The coefficients $c_{\lambda,n}^\nu$ depend on the gauge group, and for clarity we will denote those by $a_{\lambda,n}^\nu$ for $U(N)$ and by $b_{\lambda,n}^\nu$ for $SO(N)$.

\begin{rem}
In the case of $U(N)$, the above formula is an easy generalization of 
\begin{equation}
\Tr\tens U^n_\Knot=\sum_{\lambda\in\Part_n}\chi_\lambda(\mathcal C_{(n)})\Tr_{V_\lambda}\tens U_\Knot,
\end{equation}
where $\chi_\lambda$ is the character of the symmetric group $\Sym_N$ in the representation indexed by the partition $\lambda$ and $\mathcal C_{(n)}$ is the conjugacy class of one $n$-cycle in $\Sym_N$. This formula is precisely \eqref{adams_wilson} for the the fundamental representation of $U(N)$. As we will see later, the coefficients $a_{\lambda,n}^\nu$ can be expressed in terms of the characters of the symmetric group.
\end{rem}

\section{Knot Operators Formalism}
We move towards the study of Wilson loop operators associated with torus knots. The main result of this section is a formula for the matrix elements of torus knot operators that is much simpler than the one of Labastida \textit{et al.} \cite{Labastida91a}. Eventually, we will provide a simple formula for the quantum invariants of torus knots.

\subsection{Construction of the operator formalism} If a knot $\Knot$ lies on a surface $\Sigma$, the Wilson loop associated with $\Knot$ can be represented by an operator $\tens W_\lambda^\Knot$ acting on a finite-dimensional Hilbert space $\Hilb(\Sigma)$. For example, the trefoil knot pictured on \figref{trefoil} lies on the torus $\Torus^2$, and hence can be represented by an operator on $\Hilb(\Torus^2)$.

\begin{figure}[htb]
\begin{center}
\includegraphics[width=.7\linewidth]{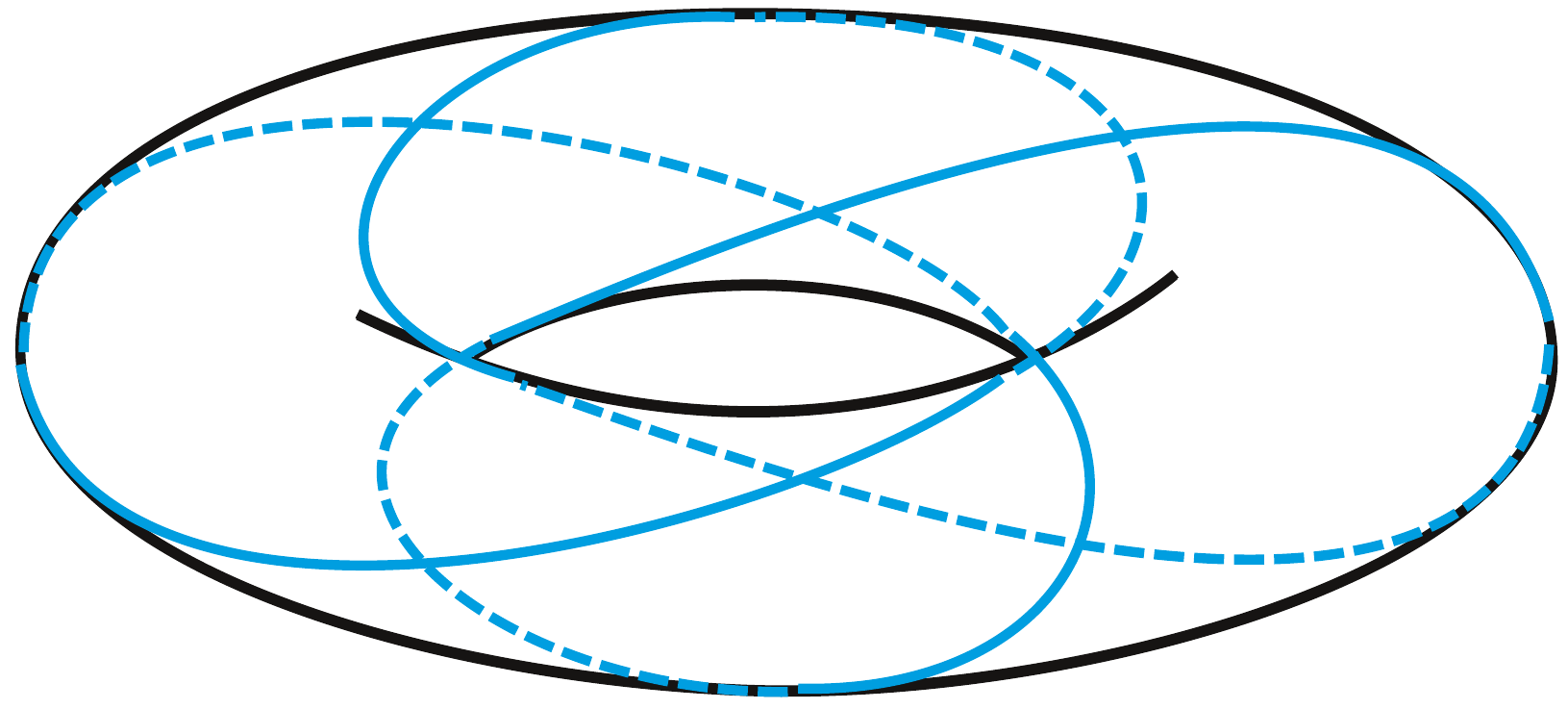}
\caption{Knot lying on a surface (torus knot).}\label{trefoil}
\end{center}
\end{figure}

In the case of torus knots, an important achievement of \cite{Labastida91a} is the construction of the operator formalism that was just alluded to. The original paper treats the case of $U(N)$ and arbitrary gauge groups are addressed in \cite{Labastida96a}. $\Hilb(\Torus^2)$ is the physical Hilbert space of Chern--Simons theory on $\R\times\Torus^2$, which is the finite-dimensional complex vector space with orthonormal basis
\begin{equation}
\label{basis}
\Big(\ket{\rho+\lambda}:\lambda\in\Lambda_W^+\Big)
\end{equation}
indexed by strongly dominant weights. Each of these states is obtained by inserting a Wilson loop in the representation $\lambda$ along the noncontractible cycle of the torus (\figref{unknot}). The state $\ket\rho$ associated with the Weyl vector $\rho$ corresponds to the vacuum (no Wilson loop inserted). 

\begin{figure}[htb]
\begin{center}
\includegraphics[width=.7\linewidth]{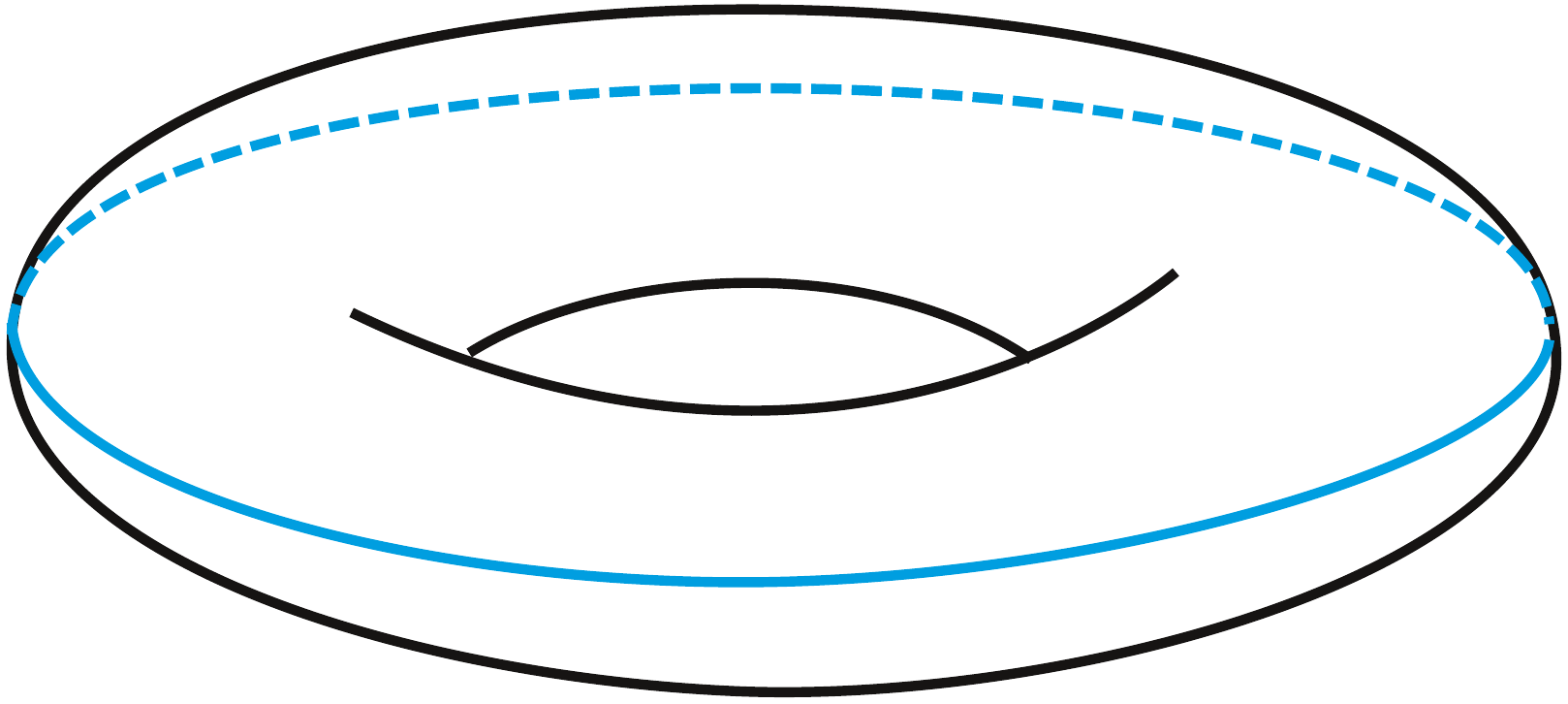}
\caption{Wilson loop $\tens W^{(1,0)}_\lambda$ around the noncontractible cycle of $\Torus^2$.}\label{unknot}
\end{center}
\end{figure}

To be more rigorous, one should restrict \eqref{basis} to integrable representations at level $k$. However, one can show that, provided $k$ is large enough, all representations that arise from the action of knot operators are integrable. Hence, we formally work as if $k$ were infinite.

We denote by $\Torus^n_m$ the $(n,m)$-torus link. $\Torus_m^n$ is a knot if and only if $n$ and $m$ are coprime. We denote by $\tens W^{(n,m)}_\lambda$ the corresponding torus knot operator. The following formula is due to \cite{Labastida91a} for the group $U(N)$, and to \cite{Labastida96a} for an arbitrary gauge group:
\begin{equation}
\label{knot_op}
\tens W_\lambda^{(n,m)}\ket p=\sum_{\mu\in M_\lambda}\exp\Big[i\pi\frac{nm}{2yk+\check c}\mu^2+2\pi i\frac{m}{2yk+\check c}p\cdot\mu\Big]\ket{p+n\mu}.
\end{equation}
In this formula, $M_\lambda$ denotes the set of weights of the irreducible $G$-module $V_\lambda$, $y$ is the Dynkin index of the fundamental representation and $\check c$ is the dual Coxeter number of $G$. The quantization condition requires that $2yk$ is an integer. 

Expression \eqref{knot_op} is actually more complicated than it seems, because not all weights $p+n\mu$ are of the form $\rho+\nu$ for some $\nu\in\Lambda_W^+$. Hence it is very difficult to get tractable formulae for $\av{\tens W_\lambda^\Knot}$ from \eqref{knot_op}. To simplify the computation of the invariants, we shall provide simple expressions for the matrix elements. This result has been established in our master's thesis \cite{Stevan09a} for the group $SU(N)$.

\subsection{Parallel cabling of the unknot}
To begin with, we consider an $n$-parallel cabling\footnote{Here parallel cabling is not to be understood in the classical sense. Usually the $n$-parallel cable of a knot is a $n$-component link, which should be represented by the product of operators $(\Tr_{V_\lambda}\tens U)^n$. In our case the $n$-parallel cable is the quantum quantity $\Tr_{V_\lambda}(\tens U^n)$.} of the unknot represented by the operator $\tens W^{(n,0)}_\lambda$. It may look a bit awkward to consider such an operator, but if we manage to cope with the exponential factor we can reduce any $\tens W^{(n,m)}_\lambda$ to $\tens W^{(n,0)}_\lambda$. From our considerations on powers of the holonomy, it is clear that
\begin{equation*}
\tens W^{(n,0)}_\lambda=\sum_{\nu\in\Lambda_W^+}c_{\lambda,n}^\nu\tens W^{(1,0)}_\nu
\end{equation*}

As a result of this operator expansion, and since $\tens W_\lambda^{(1,0)}\ket{\rho}=\ket{\rho+\lambda}$, we get the formula
\begin{equation}
\label{cabling_unknot}
\tens W^{(n,0)}_\lambda\ket\rho=\sum_{\nu\in\Lambda_W^+}c_{\lambda,n}^\nu\ket{\rho+\nu}.
\end{equation}
This equality can also be proved from the explicit representation of $\tens W^{(n,m)}_\lambda$ on $\Hilb(\Torus^2)$. More details are given in Appendix \ref{appendix_A}.

\subsection{Matrix elements of torus knot operators}To deal with the generic torus knot operator $\tens W^{(n,m)}_\lambda$, we introduce a diagonal operator
$$\tens D\ket{\rho+\lambda}=e^{2\pi i\frac m n h_{\rho+\lambda}}\ket{\rho+\lambda},$$
where
$$h_{p}=\frac{p^2-\rho^2}{2(2yk+\check c)}$$
is a conformal weight of the WZW model. The action of $\tens W_\lambda^{(n,m)}$ and $\tens W^{(n,0)}_\lambda$ on $\ket{\rho+\eta}$ differ only by an exponential factor, which is
$$ \pi i\Big[\frac{nm}{2yk+\check c}\mu^2+\frac{2m}{2yk+\check c}p\cdot\mu\Big]=\frac{m\pi i}{n(2yk+\check c)}\big[(p+n\mu)^2-p^2\big].$$
It follows immediately that 
\begin{equation}
\label{fract_twist}
\tens W^{(n,m)}_\lambda=\tens D\tens W^{(n,0)}_\lambda\tens D^{-1}.
\end{equation}
Using this result and our discussion on $\tens W_\lambda^{(n,0)}$, we obtain a simple formula for the matrix elements of $\tens W_\lambda^{(n,m)}$:
\begin{equation}
\label{vev_torus}
\tens W^{(n,m)}_\lambda\ket\rho=\sum_{\nu\in\Lambda_W^+}c_{\lambda,n}^\nu e^{2\pi i\frac m nh_{\rho+\nu}}\ket{\rho+\nu}.
\end{equation}
\begin{rem}
This formula contains the same ingredients as Lin and Zheng's formula \cite{Lin06a} for the colored HOMFLY polynomial. One of our goals was to reproduce this formula in the framework of Chern--Simons theory.
\end{rem}

\subsection{Fractional twists}
Formula \eqref{vev_torus} resembles a result of Morton and Manch\'on \cite{Morton08a} on cable knots, to which we shall return in Section \ref{cabling}. Following their terminology, we shall refer to $\tens D$ as a fractional twist. In fact, there are intrinsic reasons in Chern--Simons theory to refer to $\tens D$ as a fractional twist.

We recall that the mapping class group of the torus is $SL(2,\ZZ)$. It has two generators, $\tens T$ and $\tens S$; the former represents a Dehn twist and the later exchanges the homology cycles. There is an unitary representation $\mathcal R:SL(2,\ZZ)\longrightarrow GL\big(\Hilb(\Torus^2)\big)$ \cite{Fuchs97a}, and $\tens T$ acts by
\begin{equation*}
\mathcal R(\tens T)\ket p=e^{2\pi i(h_p+\frac{c}{12})}\ket{p}
\end{equation*}
where
$$c=\frac{2yk\dim\g}{2yk+\check c}.$$
If we redefine $\tens D$ to act as
\begin{equation*}
\tens D\ket p=e^{2\pi i\frac m n(h_p+\frac{c}{12})}\ket{p},
\end{equation*}
formula \eqref{fract_twist} remains true and we can consider $\tens D$ as the $\frac m n$-th power of $\mathcal R(\tens T)$. Furthermore $SL(2,\ZZ)$ acts by conjugation
\begin{equation}
\label{action_t}
\mathcal R(\tens M)\tens W^{(n,m)}_\lambda\mathcal R(\tens M)^{-1}=\tens W^{(n,m)\tens M}_\lambda,
\end{equation}
where $(n,m)\tens M$ stands for the natural action by right multiplication.

If we define $\tens T^{m/n}=\begin{pmatrix}
1 & \frac m n\\
0 & 1
\end{pmatrix}$ and extend $\mathcal R$ to such elements, then $\tens D=\mathcal R(\tens T^{m/n})$ and formula \eqref{fract_twist} also extends to
\begin{equation*}
\mathcal R(\tens T^{m/n})\tens W^{(n,0)}_\lambda\mathcal R(\tens T^{m/n})^{-1}=\tens W^{(n,0)\tens T^{m/n}}_\lambda=\tens W^{(n,m)}_\lambda.
\end{equation*}
With this identification it is clear why $\tens T^{m/n}$ (and its representative $\tens D$) should be called a fractional twist. It is, however, less obvious that $\mathcal R$ extends to $\tens T^{m/n}$.

\begin{rem}
Any torus knot can be obtained from the unknot by a complicated sequence of Dehn twists along both homology cycles. With a fractional twist we obtain $\Torus^n_m$ in one step from $n$-copies of the unknot. 

Our computations indicate that fractional twists have simple actions on Chern--Simons observables (at least on torus knot operators). Hopefully, fractional twists apply to more general knots.
\end{rem}

\section{Invariants of Cable Knots}\label{cabling}
We extend our analysis to cable knots from the point of view of Chern--Simons theory. Consider a knot $\Knot\subset\S^3$ and its tubular neighborhood $\T_\Knot$. Let $Q$ be a knot in the standard solid torus $\T$ and $i_\Knot:\T\into\T_\Knot$ the embedding of $\T$ into $\T_\Knot$. The satellite $\Knot*Q$ is the knot $i_\Knot(Q)$ obtained by placing $Q$ in the tubular neighborhood of $\Knot$. In case the pattern $Q$ is a torus knot, the satellite is called a cable. \figref{satellite} illustrates a cabling of the trefoil.

\begin{figure}[htb]
\begin{center}
\begin{tabular}{c@{\qquad\qquad\qquad}c}
\includegraphics[width=.3\linewidth]{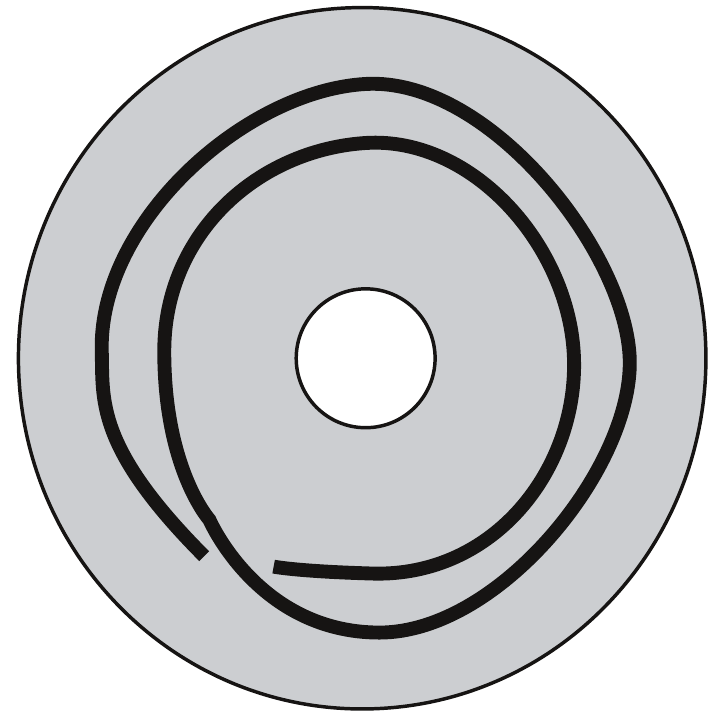} & 
\includegraphics[width=.3\linewidth]{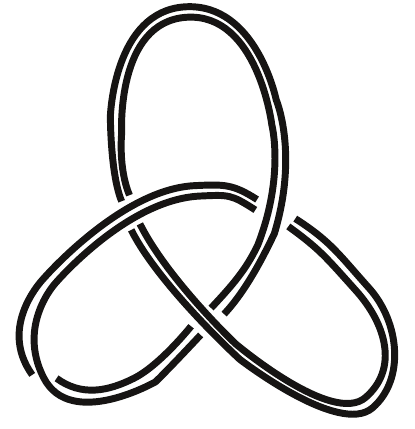}\\
$Q=\Torus^2_1$ & $\Knot*Q$
\end{tabular}
\caption{Cabling of the trefoil knot by the $(2,1)$-torus knot pattern.}
\label{satellite}
\end{center}
\end{figure}

We follow the procedure described in \cite{Witten89a}, translated in terms of knot operators. The path integral over the field configuration with support in $M'=\S^3\setminus\bar{\T_\Knot}$ gives a state
\begin{equation*}
\bra{\phi_{M'}}\in\Hilb(\partial\T_\Knot)^*,
\end{equation*}
since the boundary of $M'$ is $\partial\T_\Knot$ with the opposite orientation, and the path integral over $\T$ gives a state
\begin{equation*}
\tens W_\lambda^{(n,m)}\ket{\phi_\T}\in\Hilb(\Torus^2)
\end{equation*}
when the pattern $\Torus_m^n$ is inserted in the solid torus. The homeomorphism
$$i_\Knot|_{\Torus^2}:\Torus^2\longrightarrow\partial\T_\Knot$$
is represented by an operator $\tens F_\Knot:\Hilb(\Torus^2)\longrightarrow \Hilb(\partial\T_\Knot)$. We deduce the formula
\begin{equation*}
W_\lambda(\Knot*\Torus^n_m)=\frac{\bra{\phi_{M'}}\tens F_\Knot\tens W_\lambda^{(n,m)}\ket{\phi_\T}}{\bra{\phi_{M'}}\tens F_\Knot\ket{\phi_\T}}.
\end{equation*}
In particular, when the trivial pattern $\Torus^1_0$ is placed in the neighborhood $\T_\Knot$, the resulting satellite is $\Knot$:
\begin{equation*}
W_\lambda(\Knot)=\frac{\bra{\phi_{M'}}\tens F_\Knot\tens W_\lambda^{(1,0)}\ket{\phi_\T}}{\bra{\phi_{M'}}\tens F_\Knot\ket{\phi_\T}}.
\end{equation*}

Using our relation between $\tens W^{(n,m)}_\lambda$ and $\tens W^{(1,0)}_\lambda$, we deduce the following formula for the invariant of cable knots:
\begin{equation}
W_\lambda(\Knot*\Torus^n_m)=\sum_{\nu\in\Lambda_W^+}a_{\lambda,n}^\nu e^{-2\pi i\frac m nh_{\rho+\nu}} W_\nu(\Knot)
\end{equation}
for $U(N)$, and the same formula with $a_{\lambda,n}^\nu $ replaced by $b_{\lambda,n}^\nu $  for $SO(N)$. This formula has been proved by Morton and Manch\'on \cite{Morton08a} for HOMFLY invariants. The analogous for Kauffman invariants seems to be new.

\section{Quantum Invariants of Torus Knots}
In the preceding we have not specified the $3$-manifold $M$ onto which the knots are embedded, but the construction of the operator formalism implicitly requires $M$ to admit a genus-$1$ Heegaard splitting. The case of interest, which is $M=\S^3$, admits the decomposition into two solid tori pictured on \figref{heegaard}.

\begin{figure}[htb]
\begin{center}
\includegraphics[width=.5\linewidth]{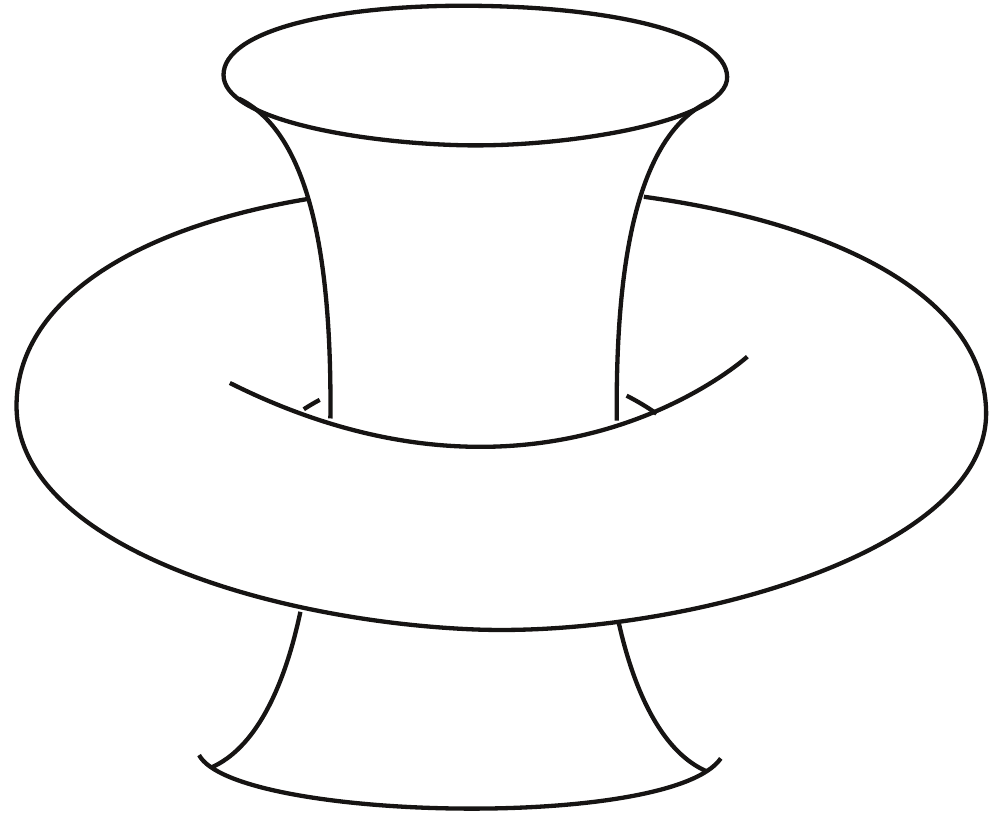}
\caption{Heegaard splitting of $\S^3$ as two solid tori.}\label{heegaard}
\end{center}
\end{figure}

The choice of a homeomorphism to glue both solid tori together determines Chern--Simons invariants through the following formula
\cite{Labastida95a}
\begin{equation}
W_\lambda(\Torus_m^n)=\frac{\bra\rho\tens{FW}^{(n,m)}\ket\rho}{\bra\rho\tens F\ket\rho},
\end{equation}
where $\tens F$ is an operator on $\Hilb(\Torus^2)$ that represents the homeomorphism. But this choice also determines a framing $w(\Knot)$ of the knot. We will correct $W_\lambda(\Knot)$ by the deframing factor $e^{-2\pi i w(\Knot)h_{\rho+\lambda}}$ \cite{Witten89a} to express the invariants in the standard framing.

It is common to glue the solid tori along the homeomorphism represented by $\tens S$ in the mapping class group (the one that exchanges the two homology cycles of $\Torus^2$). The framing determined by this choice turns out to be $mn$ for the $(n,m)$-torus knot. Its action on $\Hilb(\Torus^2)$ is given by the Kac--Peterson formula \cite{Fuchs97a}
\begin{equation}
\label{s_matrix}
\bra p\tens S\ket{p'}=\frac{i^{\abs{\Delta_+}}}{(2yk+\check c)^{1/2}}\Big|\frac{\Lambda_W}{\Lambda_R}\Big|\sum_{w\in\W}(-1)^w e^{-\frac{2\pi i}{2yk+\check c} p\cdot w(p')}.
\end{equation}

Depending on the choice of the gauge group, several invariants can be computed. Our results apply to any semisimple Lie group, but we will restrict ourselves to classical Lie groups. As it turns out, the group $U(N)$ reproduces the colored HOMFLY invariants, whereas both groups $SO(N)$ and $Sp(N)$ reproduce the colored Kauffman invariants. 
\subsection{Colored HOMFLY polynomial}
The precise relation between colored HOMFLY invariants and Chern--Simons invariants with gauge group $U(N)$ is the following:
\begin{equation}
\label{homfly}
H_\lambda^{\Knot}(t,v)=e^{-2\pi iw(\Knot)h_{\rho+\lambda}}W_\lambda(\Knot)\big|_{e^{\frac{-\pi i}{k+N}}=t,t^N=v}
\end{equation}
where $t=e^{\frac{-\pi i}{k+N}}$ and $v=t^N$ are considered as independent variables. Since $G=U(N)$ has been fixed, we have replaced $\check c$ by $N$ and $y$ by $1/2$.

We use the notation $H^{(n,m)}_\lambda$ for the HOMFLY invariants of the $(n,m)$-torus knot. It is easy to see that $e^{2\pi ih_{\rho+\lambda}}=t^{-\varkappa_\lambda}v^{-\abs\lambda}$, where $\varkappa_\lambda=\sum_{i=1}^{\ell(\lambda)}(\lambda^i-2i+1)\lambda^i$. By using the action of knot operators,
\begin{align*}
H_\lambda^{(n,m)}(t,v)&=e^{-2\pi inmh_{\rho+\lambda}}W_\lambda(\Torus_m^n)\big|_{e^{-\frac{\pi i}{k+N}}=t,t^N=v}\\
&=t^{mn\varkappa_\lambda}v^{mn\abs\lambda}\sum_{\nu\in\Lambda_W^+}a_{\lambda,n}^\nu t^{-\frac m n\varkappa_\nu}v^{-\frac m n\abs\nu} W_\nu(\Torus_0^1).
\end{align*}
The invariant of the unknot $W_\nu(\Torus_0^1)$ is called the quantum dimension of $V_\lambda$. Using the Kac--Peterson formula \eqref{s_matrix} and the Weyl character formula, one obtains
$$ W_\lambda(\Torus_0^1)=\frac{\bra\rho\tens S\ket{\rho+\lambda}}{\bra\rho\tens S\ket\rho}=\ch_\lambda\Big[ -\frac{2\pi i}{k+N}\rho\Big].$$
This expression is a function of $t$ and $v$ given by the Schur polynomial $s_\lambda(x^1,\dots,x^N)$ evaluated at $x^i=t^{N-2i+1}$. We denote this function by $s_\lambda(t,v)$.

Finally, by showing that all $\nu\in\Part$ appearing in the sum satisfy $\abs\nu=n\abs\lambda$, we obtain the following formula:
\begin{equation}
\label{homfly_torus}
H_\lambda^{(n,m)}(t,v)=t^{mn\varkappa_\lambda}v^{m(n-1)\abs\lambda}\sum_{|\nu|=n|\lambda|}a_{\lambda,n}^\nu t^{-\frac m n\varkappa_\nu}s_\nu(t,v).
\end{equation}
This formula has already been proved by Lin and Zheng \cite{Lin06a} starting from the rigorous quantum group definition. This formula is much simpler than the one originally obtained by Labastida and Mariño by using knot operators \cite{Labastida01a}.

For actual calculations the following expression is useful:
$$a_{\lambda,n}^\nu=\sum_{\mu\in\Part_{\abs\lambda}}\frac{1}{z_\mu}\chi_\lambda(\mathcal C_\mu)\chi_\nu(\mathcal C_{n\mu}).$$
It is easily proved using Frob\'enius formula for the characters of the symmetric group.

\begin{example}
Apart from the examples found in \cite{Lin06a}, we obtained for $(3,m)$-torus knots the following results:
{\small \begin{align*}
H_{\Yboxdim4pt\yng(3)}^{(3,m)}&=t^{18m}v^{6m}\Big[t^{-24m} s_{(9)}-t^{-18m}s_{(8,1)}+t^{12m} s_{(7,1^2)} \\
&\qquad\qquad+t^{-10m} s_{(6,3)}- t^{-8m}s_{(6,2,1)}-t^{-8m} s_{(5,4)}\\
&\qquad\qquad+t^{-4m}s_{(5,2^2)}+t^{-4m}s_{(4^2,1)}-t^{-2m}s_{(4,3,2)}+s_{(3^3)}\Big]
 \end{align*}
 \begin{align*}
H_{\Yboxdim4pt\yng(2,1)}^{(3,m)}&=v^{6m}\Big[t^{-10m} s_{(6,3)}- t^{-8m}s_{(6,2,1)}+t^{-6m}s_{(6,1^3)}-t^{-8m} s_{(5,4)}\\
&\qquad\qquad+t^{-4m}s_{(5,2^2)}-s_{(5,1^4)}+t^{-4m}s_{(4^2,1)}-t^{-2m}s_{(4,3,2)}\\
&\qquad\qquad+t^{6m}s_{(4,1^5)}+2s_{(3^3)}-t^{2m}s_{(3^2,2,1)}+t^{4m}s_{(3^2,1^3)}\\
&\qquad\qquad+t^{4m}s_{(3,2^3)}-t^{8m}s_{(3,2,1^4)}-t^{8m}s_{(2^4,1)}+t^{10m}s_{(2^3,1^3)}\Big]
 \end{align*}
 \begin{align*}
H_{\Yboxdim4pt\yng(1,1,1)}^{(3,m)}&=t^{-18m}v^{6m}\Big[s_{(3^3)}-t^{2m}s_{(3^2,2,1)}+t^{4m}s_{(3^2,1^3)}\\
&\qquad\qquad+t^{4m}s_{(3,2^3)}-t^{8m}s_{(3,2,1^4)}+t^{12m}s_{(3,1^6)}\\
&\qquad\qquad-t^{8m}s_{(2^4,1)}+t^{10m}s_{(2^3,1^3)}-t^{18m}s_{(2,1^7)}+t^{24m}s_{(1^9)}\Big]
\end{align*}}
\end{example}
\begin{rem}
For the sake of simplicity, we have restricted our analysis to polynomial representations of $U(N)$; analogous formulae, which will not be presented there, exist for composite representations. For example, Paul \textit{et al.} \cite{Paul10a} compute such invariants for $(2,2m+1)$-torus knots.
\end{rem}

\subsection{Colored Kauffman polynomial}
Colored Kauffman invariant are obtained from Chern--Simons theory with gauge group $SO(N)$ by
\begin{equation}
K_\lambda^{\Knot}(t,v)=e^{-2\pi i w(\Knot)h_{\rho+\lambda}}W_\lambda(\Knot)\big|_{e^{\frac{-\pi i}{2k+N-2}}=t,t^{N-1}=v}
\end{equation}
For the Lie group $SO(N)$, one has $\check c=N-2$ and $y=1$, regardless of parity.

Using the fact that $e^{2\pi ih_{\rho+\lambda}}=t^{-\varkappa_\lambda}v^{-\abs\lambda}$, the procedure is very similar to the case of $U(N)$. The quantum dimension of $V_\lambda$, which is $W_\lambda(\Torus^1_0)$, is a function of $t$ and $v$ that we denote $d_\lambda(t,v)$. Thank to Weyl character formula, it is given by the character of $SO(N)$; there are explicit expressions in \cite{Bouchard04a}. 

The final result is the exact analogous of \eqref{homfly_torus},
\begin{equation}
\label{kauffman_torus}
K_\lambda^{(n,m)}(t,v)=t^{mn\varkappa_\lambda}v^{mn\abs\lambda}\sum_{|\nu|\le n|\lambda|}b_{\lambda,n}^\nu t^{-\frac m n\varkappa_\nu}v^{-\frac m n\abs\nu}d_\nu(t,v).
\end{equation}
This formula had in fact been derived by L. Chen and Q. Chen \cite{Chen09a}; the proof is similar to \cite{Lin06a}.

The main difference, as compared with \eqref{homfly_torus}, is that the coefficients $b_{\lambda,n}^\nu$ are those of $SO(N)$, and they are nonzero also for $\abs\nu\neq n\abs\lambda$. To express these coefficients in terms of the $a_{\lambda,n}^\nu$, we use relations between characters of $SO(N)$ and $U(N)$ obtained by Littlewood \cite{Littlewood40a}. There are two formulae that give $b_{\lambda,n}^\nu$:
\begin{equation}\label{adams_so}
\begin{array}{rl}
b_{\lambda,n}^\nu&\displaystyle =\sum_{\eta\in\Part}\sum_{\mu=\bar\mu} (-1)^{\frac{\abs\mu-r(\mu)}{2}}N_{\mu\eta}^\lambda \sum_{\abs\tau=n\abs\eta} a_{\eta,n}^\tau \sum_{\xi\in\Part} \sum_{\nu\in\Part}(-1)^{\abs\xi} N_{\xi\nu}^\tau\\
& \displaystyle=\sum_{\eta\in\Part}\sum_{\gamma\in\mathscr C} (-1)^{\abs\gamma/2}N_{\gamma\eta}^\lambda\sum_{\abs\tau=n\abs\eta}a_{\eta,n}^\tau\sum_{\nu\in\Part}\sum_{\delta\in\mathscr D} N_{\delta\nu}^\tau.
\end{array}\end{equation}
More details, including notations, can be found in Appendix \ref{so_su}. In principle the first formula applies to $N$ odd and the second to $N$ even, but they seem to give the same result. A similar situation occurs for tensor products where the decomposition does not depend on the parity of $N$.

\begin{example}\label{kauff_2}
For $(2,m)$-torus knots, the colored Kauffman invariants are given by
{\small
\begin{align*}
K_{\Yboxdim4pt\yng(1)}^{(2,m)}&=v^{2m}\Big[t^{-m}v^{-m} d_{(2)}-t^{m}v^{-m}d_{(1^2)} + 1\Big] 
 \end{align*}
 \begin{multline*}
K_{\Yboxdim4pt\yng(2)}^{(2,m)}=t^{4m}v^{4m}\Big[t^{-6m}v^{-2m} d_{(4)}-t^{-2m}v^{-2m} d_{(3,1)}\\
+v^{-2m} d_{(2^2)}+t^{-m}v^{-m} d_{(2)}-t^{m}v^{-m}d_{(1^2)} + 1 \Big]
 \end{multline*}
 \begin{multline*}
K_{\Yboxdim4pt\yng(1,1)}^{(2,m)}=t^{-4m}v^{4m}\Big[v^{-2m}d_{(2^2)}-t^{2m}v^{-2m}d_{(2,1^2)}\\
+t^{6m}v^{-2m}d_{(1^4)}+t^{-m}v^{-m} d_{(2)}-t^{m}v^{-m}d_{(1^2)} + 1 \Big]
 \end{multline*}
 \begin{align*}
K_{\Yboxdim4pt\yng(3)}^{(2,m)}&=t^{12m}v^{6m}\Big[1+t^{-15m}v^{-3m}d_{(6)}-t^{-9m}v^{-3m}d_{(5,1)}\\
&\qquad+t^{-5m}v^{-3m}d_{(4,2)}-t^{-3m}v^{-3m}d_{(3,3)}+t^{-6m}v^{-2m} d_{(4)}\\
&\qquad-t^{-2m}v^{-2m} d_{(3,1)}+v^{-2m} d_{(2^2)}+t^{-m}v^{-m} d_{(2)}-t^{m}v^{-m}d_{(1^2)} \Big]
 \end{align*}
 \begin{align*}
K_{\Yboxdim4pt\yng(2,1)}^{(2,m)}&=v^{6m}\Big[1+t^{-5m}v^{-3m}d_{(4,2)}-t^{-3m}v^{-3m}d_{(4,1^2)}-t^{-3m}v^{-3m}d_{(3^2)}\\
&\qquad+t^{3m}v^{-3m}d_{(3,1^3)}+t^{3m}v^{-3m}d_{(2^3)}-t^{5m}v^{-3m}d_{(2^2,1^2)}\\
&\qquad+t^{-6m}v^{-2m} d_{(4)}-t^{-2m}v^{-2m} d_{(3,1)}+2v^{-2m} d_{(2^2)}\\
&\qquad-t^{2m}v^{-2m}d_{(2,1^2)}+t^{6m}v^{-2m}d_{(1^4)}+2t^{-m}v^{-m} d_{(2)}-2t^{m}v^{-m}d_{(1^2)}\Big]
 \end{align*}
 \begin{align*}
K_{\Yboxdim4pt\yng(1,1,1)}^{(2,m)}&=t^{-12m}v^{6m}\Big[1+t^{3m}v^{-3m}d_{(2^3)}-t^{5m}v^{-3m}d_{(2^2,1^2)}\\
&\qquad+t^{9m}v^{-3m}d_{(2,1^4)}-t^{15m}v^{-3m}d_{(1^6)}+t^{-2m}d_{(2^2)}\\
&\qquad-t^{2m}v^{-2m}d_{(2,1^2)}+t^{6m}v^{-2m}d_{(1^4)}+t^{-m}v^{-m} d_{(2)}-t^{m}v^{-m}d_{(1^2)}\Big]
\end{align*}}
\end{example}
\begin{example}\label{kauff_3}
For $(3,m)$-torus knots we further obtain
{\small
\begin{align*}
K_{\Yboxdim4pt\yng(1)}^{(3,m)}&=v^{2m}\Big[t^{-2m}d_{(3)}-d_{(2,1)} + t^{2m}d_{(1^3)}\Big] 
 \end{align*}
 \begin{multline*}
K_{\Yboxdim4pt\yng(2)}^{(3,m)}=t^{6m}v^{6m}\Big[t^{-10m}v^{-2m}d_{(6)}-t^{-6m}v^{-2m}d_{(5,1)} + t^{-2m}v^{-2m}d_{(4,1^2)}\\
 + t^{-2m}v^{-2m}d_{(3^2)}-v^{-2m}d_{(3,2,1)}+ t^{2m}v^{-2m}d_{(2^3)}+1\Big] 
 \end{multline*}
  \begin{multline*}
K_{\Yboxdim4pt\yng(1,1)}^{(3,m)}=t^{-6m}v^{6m}\Big[t^{-2m}v^{-2m}d_{(3^2)}-v^{-2m}d_{(3,2,1)}+ t^{2m}v^{-2m}d_{(3,1^3)}\\
 + t^{2m}v^{-2m}d_{(2^3)}-t^{6m}v^{-2m} d_{(2,1^4)}+t^{10m}v^{-2m}d_{(1^6)}+1\Big]
 \end{multline*}}
\end{example}

\begin{rem}
These results are rather simple as compared with formula \eqref{adams_so} for the Adams coefficients. We observed important cancellations of terms; thus it might be possible to simplify \eqref{adams_so}. In particular, Kauffman invariants present the following recursive structure: $K_{\Yboxdim4pt\yng(1)}$ appears in $K_{\Yboxdim4pt\yng(2)}$, $K_{\Yboxdim4pt\yng(2)}$ appears in turn in $K_{\Yboxdim4pt\yng(3)}$, and so on.
\end{rem}

\section{Quantum Invariants of Torus Links}
The formulae for HOMFLY and Kauffman invariants generalizes to links by using the fusion rule \eqref{fusion} and taking into account the framing correction. One obtains
\begin{gather}
H_{\lambda_1,\dots,\lambda_L}^{(Ln,Lm)}=t^{mn\sum_{\alpha=1}^L\varkappa_{\lambda_\alpha}}\sum_{\mu\in\Part} N_{\lambda_1,\dots,\lambda_L}^\mu t^{-mn\varkappa_\mu}H_\mu^{(n,m)}\\
K_{\lambda_1,\dots,\lambda_L}^{(Ln,Lm)}=t^{mn\sum_{\alpha=1}^L\varkappa_{\lambda_\alpha}} v^{\sum_{\alpha=1}^L mn\abs{\lambda_\alpha}} \sum_{\mu\in\Part} M_{\lambda_1,\dots,\lambda_L}^\mu t^{-mn\varkappa_\mu}v^{-m n\abs\mu} K_\mu^{(n,m)}\nonumber
\end{gather}
for the $(Ln,Lm)$-torus link. The first formula is equivalent to the formula of \cite{Lin06a} for torus links.
 
\begin{example}
{\small
For $(4,2m)$-torus links, the colored Kauffman invariants are
\begin{multline*}
 K_{\Yboxdim4pt\yng(1),\yng(1)}^{(4,2m)}=v^{4m}\Big[3+t^{-6m}v^{-2m}d_{(4)}-t^{-2m}v^{-2m}d_{(3,1)}+2v^{-2m}d_{(2^2)}\\
-t^{2m}v^{-2m}d_{(2,1^2)}+t^{6m}v^{-2m}d_{(1^4)}+2t^{-m}v^{-m}d_{(2)}-2t^{m}v^{-m}d_{(1^2)}\Big]
 \end{multline*}
\begin{align*}
 K_{\Yboxdim4pt\yng(2),\yng(1)}^{(4,2m)}&=t^{4m}v^{6m}\Big[t^{-15m}v^{-3m}d_{(6)}-t^{-9m}v^{-3m}d_{(5,1)}+2t^{-5m}v^{-3m}d_{(4,2)}\\
 &\qquad\qquad-t^{-3m}v^{-3m}d_{(4,1^2)}-2t^{-3m}v^{-3m}d_{(3^2)}+t^{3m}v^{-3m}d_{(3,1^3)}\\
 &\qquad\qquad+t^{3m}v^{-3m}d_{(2^3)}-t^{5m}v^{-3m}d_{(2^2,1^2)}+2t^{-6m}v^{-2m}d_{(4)}\\
 &\qquad\qquad-2t^{-2m}v^{-2m}d_{(3,1)}+3v^{-2m}d_{(2^2)}-t^{2m}v^{-2m}d_{(2,1^2)}\\
 &\qquad\qquad+t^{6m}v^{-2m}d_{(1^4)}+4t^{-m}v^{-m}d_{(2)}-4t^mv^{-m}d_{(1^2)}+3\Big]
 \end{align*}
 \begin{align*}
  K_{\Yboxdim4pt\yng(1,1),\yng(1)}^{(4,2m)}&=t^{-4m}v^{6m}\Big[t^{-5m}v^{-3m}d_{(4,2)}-t^{-3m}v^{-3m}d_{(4,1^2)}+t^{-3m}v^{-3m}d_{(3^2)}\\
  &\qquad\qquad+t^{3m}v^{-3m}d_{(3,1^3)}+2t^{3m}v^{-3m}d_{(2^3)}-2t^{5m}v^{-3m}d_{(2^2,1^2)}\\
  &\qquad\qquad+t^{9m}v^{-3m}d_{(2,1^4)}-t^{15m}v^{-3m}d_{(1^6)}+t^{-6m}v^{-2m}d_{(4)}\\
  &\qquad\qquad-t^{-2m}v^{-2m}d_{(3,1)}+3v^{-2m}d_{(2^2)}-2t^{2m}v^{-2m}d_{(2,1^2)}\\
  &\qquad\qquad+t^{6m}v^{-2m}d_{(1^4)}+4t^{-m}v^{-m}d_{(2)}-4t^{m}v^{-m}d_{(1^2)}+3\Big]
 \end{align*}}
 \end{example}

\section{Mariño Conjecture for the Kauffman Invariants}
Many highly nontrivial properties of the Kauffman invariants as well as their relation to the HOMFLY invariants might be explained by a conjecture of Mariño \cite{Marino09c} that completes the prior partial conjecture of Bouchard, Florea and Mariño \cite{Bouchard05a}. This new conjecture is similar to the Labastida-Mariño-Ooguri-Vafa conjecture \cite{Ooguri00a,Labastida00a} for HOMFLY invariants, but it applies to Kauffman invariants and HOMFLY invariants with composite representations.

\subsection{Statement of the conjecture}
The conjecture contains two distinct statements, one for HOMFLY invariants including composite representations and one for both Kauffman and HOMFLY invariants. We first construct the generating functions
\begin{gather*}
Z_H(\Lin)=\sum_{\tiny \begin{matrix}\lambda_1,\dots,\lambda_L\\\mu_1,\dots,\mu_L\end{matrix}} H_{[\lambda_1,\mu_1],\dots,[\lambda_L,\mu_L]}^\Lin(t,v) s_{\lambda_1}(\vec x_1)s_{\mu_1}(\vec x_1)\cdots s_{\lambda_L}(\vec x_L)s_{\mu_L}(\vec x_L)\\
 Z_K(\Lin)=\sum_{\lambda_1,\dots,\lambda_L}K_{\lambda_1,\dots,\lambda_L}^\Lin(t,v) s_{\lambda_1}(\vec x_1)\cdots s_{\lambda_L}(\vec x_L),
 \end{gather*}
where all sums run over partitions including the empty one. The reformulated invariants $h_{\lambda_1,\dots,\lambda_L}(t,v)$ and $g_{\lambda_1,\dots,\lambda_L}(t,v)$ are defined by
\begin{gather}
\log Z_H=\sum_{d=1}^\infty \sum_{\lambda_1,\dots,\lambda_L}h_{\lambda_1,\dots,\lambda_L}(t^d,v^d)s_{\lambda_1}(\vec x_1^d)\cdots s_{\lambda_L}(\vec x_L^d)\\
\log Z_K-\frac 1 2\log Z_H=\sum_{d\text{ odd}}\sum_{\lambda_1,\dots,\lambda_L} g_{\lambda_1,\dots,\lambda_L}(t^d,v^d)s_{\lambda_1}(\vec x_1^d)\cdots s_{\lambda_L}(\vec x_L^d).\nonumber
\end{gather}
All reformulated invariants can be expressed in terms of the original invariants through computing connected vacuum expectation values, following the procedure of \cite{Labastida02a}. We suggest an alternative procedure in Appendix \ref{comput_reform}. For a knot, the lowest-order invariants are 
\begin{align*}
g_{\Yboxdim4pt\yng(1)}(t,v)&=K_{\Yboxdim4pt\yng(1)}(t,v)-H_{\Yboxdim4pt\yng(1)}(t,v)\\
g_{\Yboxdim4pt\yng(2)}(t,v)&=K_{\Yboxdim4pt\yng(2)}(t,v)-\frac 1 2K_{\Yboxdim4pt\yng(1)}(t,v)^2-H_{\Yboxdim4pt\yng(2)}(t,v)+H_{\Yboxdim4pt\yng(1)}(t,v)^2-\frac 1 2H_{[\Yboxdim4pt\yng(1),\Yboxdim4pt\yng(1)]}(t,v)\\
g_{\Yboxdim4pt\yng(1,1)}(t,v)&=K_{\Yboxdim4pt\yng(1,1)}(t,v)-\frac 1 2K_{\Yboxdim4pt\yng(1)}(t,v)^2-H_{\Yboxdim4pt\yng(1,1)}(t,v)+H_{\Yboxdim4pt\yng(1)}(t,v)^2-\frac 1 2H_{[\Yboxdim4pt\yng(1),\Yboxdim4pt\yng(1)]}(t,v).
\end{align*}
More examples can be found in \cite{Marino09c}. We now introduce the block-diagonal matrix $M_{\lambda\mu}$, which is
$$M_{\lambda\mu}(t)=\sum_{\nu\in\Part_n} \chi_\lambda(\mathcal C_\nu)\chi_\mu(\mathcal C_\nu)\prod_{i=1}^n (t^{\nu^i}-t^{-\nu^i})$$
for $\abs\lambda=\abs\mu=n$ and zero otherwise. We finally define
\begin{gather}
\begin{array}{c}
\displaystyle\hat h_{\lambda_1,\dots,\lambda_L}(t,v)=\sum_{\mu_1,\dots,\mu_L} M^{-1}_{\lambda_1\mu_1}(t)\cdots M^{-1}_{\lambda_L\mu_L}(t) h_{\mu_1,\dots,\mu_L}(t,v)\\
\displaystyle\hat g_{\lambda_1,\dots,\lambda_L}(t,v)=\sum_{\mu_1,\dots,\mu_L} M^{-1}_{\lambda_1\mu_1}(t)\cdots M^{-1}_{\lambda_L\mu_L}(t) g_{\mu_1,\dots,\mu_L}(t,v).
\end{array}
\end{gather}
The conjecture states that
$$\hat h_{\lambda_1,\dots,\lambda_L}\in z^{L-2}\ZZ[z^2,v^{\pm 1}]\qquad\text{and}\qquad\hat g_{\lambda_1,\dots,\lambda_L}\in z^{L-1}\ZZ[z,v^{\pm 1}],$$
with $z=t-t^{-1}$. In other words, there exist integer invariants $\mathcal N_{\lambda_1,\dots,\lambda_L;g,Q}^{c}$ ($c=0,1,2$) such that
\begin{equation}
\hat h_{\lambda_1,\dots,\lambda_L}(z,v)=z^{L-2}\sum_{g\ge 0}\sum_{Q\in\ZZ} \mathcal N_{\lambda_1,\dots,\lambda_L;g,Q}^0z^{2g-1}v^Q
\end{equation}
and
\begin{equation*}
\hat g_{\lambda_1,\dots,\lambda_L}(z,v)=z^{L-1}\sum_{g\ge 0}\sum_{Q\in\ZZ}\Big[ \mathcal N_{\lambda_1,\dots,\lambda_L;g,Q}^1 z^{2g}v^Q+\mathcal N_{\lambda_1,\dots,\lambda_L;g,Q}^2 z^{2g+1}v^Q\Big].
\end{equation*}

\subsection{Direct computations}We now proceed to various tests of the conjecture for torus knots and links using formulae \eqref{homfly_torus} and \eqref{kauffman_torus}. Unfortunately, we cannot test the conjecture for all torus knots at once, and since the complexity increases rapidly, only the cases $(2,m)$ and $(3,m)$ are tractable.

In principle the integer invariants can be computed as functions of $m$ (though they are in infinite number if $m$ is not fixed). In practice, however, we had to fix $m$ to obtain results in a reasonable amount of time. We have obtained generic results in a few cases, to which we shall return later on.

For $(2,m)$-torus knots, we have checked the conjecture for various values of $m$ and for several low-dimensional representations. Most of these tests had already been made by \cite{Marino09c}, using the formulae of \cite{Borhade05a} for Kauffman invariants. Recently, analogous tests have also been made for this class of knots with nontrivial framing \cite{Paul10a}.

For $(3,m)$-torus knots, we were able to verify parts of the conjecture. As an illustration, we have compiled the integer invariants $\mathcal N_{\Yboxdim4pt\yng(2),g,Q}^{1}$ of the $(3,4)$-torus knot in \tabref{tab2}.

We further have proceeded to nontrivial checks of the conjecture for $(2,2m)$- and $(4,2m)$-torus links. For definiteness we consider here the two-component trefoil link $\Torus^4_6$. We have obtained
{\small \begin{align*}\hat g_{\Yboxdim3pt\yng(1),\yng(1)}&=(36v^9-180v^7+288v^5-144v^3)z+(57v^9-453v^7+912v^5-516v^3)z^3\\
&+(36v^9-494v^7+1286v^5-828v^3)z^5+(10v^9-286v^7+1001v^5-725v^3)z^7\\
&+(v^9-91v^7+455v^5-365v^3)z^9-(15v^7-120v^5+105v^3)z^{11}\\
&-(v^7-17v^5+16v^3)z^{13}+(v^5-v^3)z^{15},
\end{align*}}from which the integer invariants can be read. We have also compiled the invariants $\mathcal N_{\Yboxdim4pt\yng(2),\yng(1);g,Q}^{2}$ of the same link in \tabref{tab3}.

It is interesting to remark that in the above formula all $\mathcal N_{\Yboxdim4pt\yng(1),\yng(1);g,Q}^{2}$ vanish. For torus knots it is the case that $\mathcal N_{\Yboxdim4pt\yng(1),g,Q}^{2}=0$, because of Labastida-Pérez relation \cite{Labastida96a}
\begin{equation*}
\frac 1 2\big[K^{(n,m)}_{\Yboxdim4pt\yng(1)}(z,v)-K^{(n,m)}_{\Yboxdim4pt\yng(1)}(-z,v)\big]=H^{(n,m)}_{\Yboxdim4pt\yng(1)}(z,v)
\end{equation*}
between the HOMFLY and the Kauffman polynomials. But this relation does not hold in for torus links, and we suggest that an appropriate generalization is
\begin{equation}
\frac{1}{2}\Big[K^{(2n,2m)}_{\Yboxdim4pt\yng(1),\yng(1)}+\bar K^{(2n,2m)}_{\Yboxdim4pt\yng(1),\yng(1)}\Big]-K^{(n,m)}_{\Yboxdim4pt\yng(1)}\bar K^{(n,m)}_{\Yboxdim4pt\yng(1)}=H^{(2n,2m)}_{\Yboxdim4pt[\yng(1),\emptyset],[\yng(1),\emptyset]}+H^{(2n,2m)}_{\Yboxdim4pt[\yng(1),\emptyset],[\emptyset,\yng(1)]}
\end{equation}
for two-components torus links, where the bar stands for the substitution $z\to -z$. More generally, we are led to conjecture that $\mathcal N_{\Yboxdim4pt\yng(1),\dots,\yng(1);g,Q}^{2}=0$ for any torus link.
 
We return to the computation of the integer invariants as functions of $m$. Formally $\mathcal N_{\lambda,g,Q}^c$ is a polynomial in $m$ with rational coefficients, enjoying the following properties: for each $m$ such that $\gcd(n,m)=1$,
\begin{enumerate}
  \item $\mathcal N_{\lambda,g,Q}^c$ is an integer;
  \item $\mathcal N_{\lambda,g,Q}^c$ vanishes for large $g$ and large $|Q|$.
\end{enumerate}
For the $(2,m)$-torus knot we were able to perform the computation for the representation $\Yboxdim4pt\yng(2)$ and for $g=0,1,2$. The results are compiled in \tabref{tab1}. The fact that these complicated expressions are indeed integers is not completely trivial: let us show for instance that
$$\mathcal N_{\Yboxdim4pt\yng(2),2,3m}^1=\frac{m^2(m^2-1)(2m+1)(339m^2+296m-259)}{5760}\in\ZZ.$$
Let $p(m)=339m^2+296m-259$. We test the divisibility of the numerator by $5760=2^7\cdot 3^2\cdot 5$ for $m$ odd.
\begin{enumerate}
\item Divisibility by $5$: since $p(m)\equiv 4m^2+m+1\pmod 5$, we see that $\{m,m-1,2m+1,p(m),m+1\}$ always contains a multiple of $5$.
\item Divisibility by $3^2$: we observe that $p(m)\equiv 2m+2\pmod 3$, hence both sets $\{m,2m+1,p(m)\}$ and $\{m,m-1,m+1\}$ contain a multiple of $3$.
\item Divisibility by $2^7$: one has to consider classes modulo $16$, in particular $p(m)\equiv 3m^2+8m+13\pmod{16}$. For $m\equiv 1\pmod{8}$, we have two multiples of $8$ ($m-1$ and $p(m)$). Similarly for $m\equiv 7\pmod{8}$. In both cases there is an additional even factor ($m+ 1$ resp. $m-1$). If now $m\equiv 3\pmod{8}$, then $p(m)$ is a multiple of $16$. Also $m+1$ is a multiple of $4$, and $m-1$ is even. Similarly for $m\equiv 5\pmod{8}$.
\end{enumerate}

\begin{table}
{\small \begin{tabular}{|c|c|c|c|c|c|c|}
\hline
$\mathcal N_{\Yboxdim4pt\yng(2),g,Q}^1$ & $Q=11$ & $Q=13$ & $Q=15$ & $Q=17$ & $Q=19$ & $Q=21$\\
\hline\hline
$g=0$ & $-750$ & $3300$ & $- 5590$ & $4470$ & $- 1620$ & $190$\\
$g=1$ &$-5425$ & $27200$ & $- 49845$ & $40925$ & $- 14100$ & $1245$ \\
$g=2$ & $-17325$ & $103245$ & $- 208513$ & $176489$ & $- 57299$ & $3403$\\
$g=3$ & $-32020$ & $233835$ & $- 525576$ & $457606$ & $- 138841$ & $4996$\\
$g=4$ & $-37920$ & $348942$ & $- 880083$ & $785953$ & $- 221259$ & $4367$\\
$g=5$ & $-30177$ & $360999$ & $- 1031637$ & $942490$ & $- 244055$ & $2380$\\
$g=6$ & $-16472$ & $266337$ & $- 873189$ & $814080$ & $- 191572$ & $816$\\
$g=7$ & $-6175$ & $142083$ & $- 543170$ & $515506$ & $- 108415$ & $171$\\
$g=8$ & $-1561$ & $54921$ & $- 250153$ & $241067$ & $- 44294$ & $20$\\
$g=9$ & $-254$ & $15227$ & $- 85099$ & $83052$ & $- 12927$ &  $1$\\
$g=10$ & $-24$ & $2950$ & $- 21102$ & $20801$ & $-2625$ & \\
$g=11$ & $-1$ & $379$ & $- 3707$ & $3681$ & $- 352$ &\\
$g=12$ &  & $29$ & $- 437$ & $436$ & $-28$ & \\
$g=13$ & & $1$ & $-31$ & $31$ & $-1$ & \\
$g=14$ & & & $-1$ & $1$ & & \\
\hline
\end{tabular}\vspace{.5cm}
\caption{Integer invariants for the $(3,4)$-torus knot.}\label{tab2}}
\end{table}
\begin{table}
{\small \centering
 
  \begin{tabular}{|c|c|c|c|c|c|}
 \hline
 $\mathcal N^2_{\Yboxdim4pt\yng(2),\yng(1);g,Q}$ & $Q=7$ & $Q=9$ & $Q=11$ & $Q=13$ & $Q=15$\\
 \hline\hline
$g=0$ & $1512$ & $- 5292$ & $6804$ & $- 3780$ & $756$\\
$g=1$ & $10206$ & $- 35847$ & $44037$ & $-  22113$ & $3717$\\
$g=2$ & $30177$ & $- 108507$ & $127764$ & $- 57204$ & $7770$\\
$g=3$ & $51554$ & $- 193977$ & $220023$ & $-86738$ & $9138$\\
$g=4$ & $56536$ & $- 227868$ & $250418$ & $- 85792$ & $6706$\\
$g=5$ & $41817$ & $- 185180$ & $198272$ & $- 58102$ &  $3193$\\
$g=6$ & $21318$ & $- 106758$ & $111925$ & $ - 27472$ & $987$\\
$g=7$ & $7505$ & $- 44024$ & $45393$ & $- 9065$ & $191$\\
$g=8$ & $1792$ & $-12902$ & $13135$ & $- 2046$ & $21$\\
$g=9$ & $277$ & $- 2624$ & $2647$ & $- 301$ & $1$\\
$g=10$ & $25$ & $- 352$ & $353$ & $- 26$ & \\
$g=11$ & $1$ & $- 28$ & $28$ & $-1$ & \\    
$g=12$ & & $-1$ & $1$ &  & \\
\hline
 \end{tabular}\vspace{0.5cm}
 \caption{Integer invariants for the $(4,6)$-torus link}\label{tab3}}
 \end{table}
 
 \textbf{Acknowledgments.} We would like to thank Marcos Mariño for suggesting the subject of our master's thesis, for helpful discussions on various topics, and for comments on the manuscript. We also thank Andrea Brini for helpful discussions on large-$N$ duality and matrix models.

\appendix\section{Action of the Knot Operators on $\Hilb(\Torus^2)$}\label{appendix_A}
This appendix is devoted to the proof of formula \eqref{cabling_unknot} for the action of $\tens W^{(n,0)}_\lambda$ on $\ket\rho$. Though it can be deduced from generic considerations on Wilson loops, we provide an alternative derivation starting from the action of torus knot operators on $\Hilb(\Torus^2)$.

Our considerations are based on the following remark: the basis elements of $\Hilb(\Torus^2)$ are anti-symmetrized sums over the Weyl group
\begin{equation}
\ket{p}=\sum_{w\in\W}(-1)^w f^{w(p)},
\end{equation}
where $f^p$ is some complex function that admits a Fourier series expansion \cite{Labastida91a}. Hence we can work with the formal anti-symmetric elements
$$A_p=\sum_{w\in\W}(-1)^w e^{w(p)}$$
in $\ZZ[\Lambda_W]$ and translate the results to $\Hilb(\Torus^2)$.

We derive the required formula
\begin{equation}
\sum_{\mu\in M_\lambda} \ket{\rho+n\mu}=\sum_{\nu\in\Lambda_W}c_{\lambda,n}^\nu\ket\nu
\end{equation}
from simple properties of the Weyl group and of the weight lattice.

\begin{lemm}
The following equality holds in $\ZZ[\Lambda_W]$:
$$\sum_{\mu\in M_\lambda} A_{\rho+n\mu}=\sum_{\nu\in\Lambda_W}c_{\lambda,n}^\nu A_{\rho+\nu},$$
where $c_{\lambda,n}^\nu$ are the coefficient of the Adams operation \eqref{adams}.
\end{lemm}
\begin{proof}
Using the fact that the set of weights is just permuted by the Weyl group, we immediately obtain
\begin{align*}\sum_{\mu\in M_\lambda} A_{\rho+n\mu}&=\sum_{\mu\in M_\lambda} \sum_{w\in\W} (-1)^w e^{w(\rho+n\mu)}=\sum_{\mu\in M_\lambda}e^{n\mu}\sum_{w\in\W} (-1)^w e^{w(\rho)}\\
&=(\Psi_n\ch_\lambda)A_\rho=\sum_{\nu\in\Lambda_W}c_{\lambda,n}^\nu \ch_\nu A_\rho
\end{align*}
and the conclusion follows from Weyl character formula.
\end{proof}

Some further properties of Wilson loops can be checked explicitly for torus knot operators using similar arguments \cite{Stevan09a}.

\section{Computation of the Reformulated Invariants}\label{comput_reform}
In this appendix we give explicit formulae for the reformulated invariants $h_\lambda(t,v)$ and $g_\lambda(t,v)$. Since we shall be dealing with finite collections of all different partitions, it is convenient to introduce the set $\N[\Part]$ of finitely-supported functions $\Part\longrightarrow\N$. If we use elementary functions
$$\begin{array}{rccc}
e_\lambda: & \Part & \longrightarrow & \N\\
& \mu & \longmapsto & \delta_{\lambda\mu}
\end{array},$$
each $\boldsymbol\Lambda\in\N[\Part]$ can be written as
$$\boldsymbol\Lambda=\sum_{\lambda\in\Part}n_{\boldsymbol\Lambda}(\lambda) e_\lambda,$$
where $\vec n_{\boldsymbol\Lambda}=\big(n_{\boldsymbol\Lambda}(\lambda)\big)_{\lambda\in\Lambda}$ is a sequence with finite support. Let also $\abs{\vec n}=\sum_{\lambda\in\Part}n_{\boldsymbol\Lambda}(\lambda)$ and
$$\norm{\boldsymbol\Lambda}=\sum_{\lambda\in\Part}n_{\boldsymbol\Lambda}(\lambda)\abs\lambda.$$
We introduce the following combinatoric object: $N_{\boldsymbol\Lambda}^\eta$ is defined as
$$\prod_{\lambda\in\Part}\ch_\lambda^{n_{\boldsymbol\Lambda}(\lambda)}=\sum_{\eta\in\Part}N_{\boldsymbol\Lambda}^\eta \ch_\eta.$$
Clearly, the above sum is finite and only runs on elements such that $\abs\eta=\norm{\boldsymbol\Lambda}$.


Because of composite representations, we also need two-variables polynomials $\N[\Part,\Part]$. Introducing the elementary functions
$$\begin{array}{rccc}
e_{\lambda,\mu}: & \Part\times\Part & \longrightarrow & \N\\
& (\alpha,\beta) & \longmapsto & \delta_{\lambda\alpha}\delta_{\mu\beta}
\end{array},$$
we can write $\boldsymbol\Lambda\in\N[\Part,\Part]$ as
$$\boldsymbol\Lambda=\sum_{\lambda,\mu\in\Part}n_{\boldsymbol\Lambda}(\lambda,\mu)e_{\lambda,\mu}.$$
We define as before
$$\norm{\boldsymbol\Lambda}=\sum_{\lambda,\mu\in\Part}\big(n_{\boldsymbol\Lambda}(\lambda,\mu)+n_{\boldsymbol\Lambda}(\mu,\lambda)\big)\abs\lambda$$
and $N_{\boldsymbol\Lambda}^\eta$ by
$$\prod_{\lambda,\mu\in\Part}(\ch_\lambda\ch_\mu)^{n_{\boldsymbol\Lambda}(\lambda,\mu)}=\sum_{\eta\in\Part}N_{\boldsymbol\Lambda}^\eta \ch_\eta.$$
We write $d\vert\lambda$ if $d$ divides $\abs{\lambda}$, and we let $\mu(d)$ be the M\"obius function.

By expanding the logarithm in series, we obtained the following formulae:
{\small \begin{align*} h_\lambda &= \sum_{d|\lambda}\frac{\mu(d)}{d}\sum_{\eta\in\Part_{\abs\lambda/d}} a_{\eta,d}^\lambda\sum_{\kappa_1,\kappa_2\in\Part}N_{\kappa_1\kappa_2}^\eta  \sum_{\substack{\boldsymbol\Lambda\in\N[\Part]\\ \norm{\boldsymbol\Lambda}=\abs{\kappa_1}}}\sum_{\substack{\boldsymbol\Gamma\in\N[\Part,\Part]\\ \norm{\boldsymbol\Gamma}=\abs{\kappa_2}}}2^{\abs{\vec n_{\boldsymbol\Lambda}}}\frac{(-1)^{\abs{\vec n_{\boldsymbol\Lambda}}+\abs{\vec n_{\boldsymbol\Gamma}}+1}}{\abs{\vec n_{\boldsymbol\Lambda}}+\abs{\vec n_{\boldsymbol\Gamma}}}\\
&\qquad\times\binom{\abs{\vec n_{\boldsymbol\Lambda}}+\abs{\vec n_{\boldsymbol\Gamma}}}{\vec n_{\boldsymbol\Lambda}\quad\vec n_{\boldsymbol\Gamma}} N_{\boldsymbol\Lambda}^{\kappa_1} N_{\boldsymbol\Gamma}^{\kappa_2} \prod_{\alpha\in\Part} H_{\alpha}(t^d,v^d)^{n_{\boldsymbol\Lambda}(\alpha)}  \prod_{\beta,\gamma\in\Part} H_{[\beta,\gamma]}(t^d,v^d)^{n_{\boldsymbol\Gamma}(\beta,\gamma)}  \end{align*}}
and
{\small \begin{align*} g_\lambda&= \sum_{\text{odd }d|\lambda}\frac{\mu(d)}{d}\sum_{\eta\in\Part_{\abs\lambda/d}}a_{\eta,d}^{\lambda}\sum_{\norm{\boldsymbol\Lambda}=\abs\eta} \frac{(-1)^{\abs{\vec n_{\boldsymbol\Lambda}}-1}}{\abs{\vec n_{\boldsymbol\Lambda}}}\binom{\abs{\vec n_{\boldsymbol\Lambda}}}{\vec n_{\boldsymbol\Lambda}} N_{\boldsymbol\Lambda}^{\eta} \prod_{\alpha\in\Part}K_\alpha(t^d,v^d)^{n_{\boldsymbol\Lambda}(\alpha)}\\
&-\sum_{\text{odd }d|\lambda}\frac{\mu(d)}{d}\sum_{\eta\in\Part_{\abs\lambda/d}} a_{\eta,d}^\lambda\sum_{\kappa_1,\kappa_2\in\Part}N_{\kappa_1\kappa_2}^\eta \sum_{\substack{\boldsymbol\Lambda\in\N[\Part]\\ \norm{\boldsymbol\Lambda}=\abs{\kappa_1}}}\sum_{\substack{\boldsymbol\Gamma\in\N[\Part,\Part]\\ \norm{\boldsymbol\Gamma}=\abs{\kappa_2}}}  \frac{(-1)^{\abs{\vec n_{\boldsymbol\Lambda}}+\abs{\vec n_{\boldsymbol\Gamma}}+1}}{\abs{\vec n_{\boldsymbol\Lambda}}+\abs{\vec n_{\boldsymbol\Gamma}}}\\
&\times2^{\abs{\vec n_{\boldsymbol\Lambda}}-1} \binom{\abs{\vec n_{\boldsymbol\Lambda}}+\abs{\vec n_{\boldsymbol\Gamma}}}{\vec n_{\boldsymbol\Lambda}\quad\vec n_{\boldsymbol\Gamma}} N_{\boldsymbol\Lambda}^{\kappa_1} N_{\boldsymbol\Gamma}^{\kappa_2} \prod_{\alpha\in\Part} H_{\alpha}(t^d,v^d)^{n_{\boldsymbol\Lambda}(\alpha)}  \prod_{\beta,\gamma\in\Part} H_{[\beta,\gamma]}(t^d,v^d)^{n_{\boldsymbol\Gamma}(\beta,\gamma)}
\end{align*}}

\section{Characters of $SO(N)$}\label{so_su}
The characters of $SO(2r+1)$ and $SO(2r)$ can be represented by symmetric polynomials in $\ZZ[x_1,\dots,x_r,x_1^{-1},\dots,x_r^{-1}]$, whose explicit expression are given in \cite{Fulton91a}. They can be expressed as linear combination of Schur functions in $2r$ variables. The relations are \cite{Littlewood40a}
\begin{equation}\begin{array}{c}
\ch_\lambda^{\so(2r+1)}\displaystyle=\sum_{\eta\in\Part}\sum_{\mu=\bar\mu}(-1)^{\frac{\abs\mu-r(\mu)}{2}} N_{\mu\eta}^\lambda s_\eta
\\
\ch_\lambda^{\so(2r)}\displaystyle=\sum_{\eta\in\Part}\sum_{\gamma\in\mathscr C} (-1)^{\abs\gamma/2}N_{\gamma\eta}^\lambda s_\eta
\end{array}\end{equation}
and the reciprocals
\begin{equation}\begin{array}{c}
s_\lambda\displaystyle=\sum_{\eta\in\Part}\sum_{\xi\in\Part\cup\{\emptyset\}} (-1)^{\abs\xi/2}N_{\xi\eta}^\lambda \ch_\eta^{\so(2r+1)}\\
s_\lambda\displaystyle=\sum_{\eta\in\Part}\sum_{\delta\in\mathscr D} N_{\delta\eta}^\lambda\ch^{\so(2r)}_\eta.
\end{array}\end{equation}
In these formulae, $\bar\mu$ is the partition conjugate to $\mu$, $r(\mu)$ is the rank of $\mu$, $\mathscr C$ is the set of partitions of the form $(b_1+1,b_2+1,\dots|b_1,b_2,\dots)$ in Frob\'enius notation and $\mathscr D$ is the set of partitions into even parts only. Both sets include the empty partition, and so does the sum over self-conjugate partitions.

\begin{table}[p]
{\small \centering
\begin{tabular}{|c|c|c|}
\hline
$g$ & $Q$ & $\mathcal N_{\Yboxdim4pt\yng(2),g,Q}^0$\\
\hline\hline
 & $2m$ & $\frac{m}{2}(m^2-1)(m^2+m+4)$\\
& $2m\pm 2$ & $\frac{m^2}{3}(m^3+m^2+2m-1)$\\
 0& $2m\pm 4$ & $\frac{m^2}{12}(m^2-1)^2$\\
 & $4m$ & $1$\\
 \hline\hline
 & $2m$ & $\frac{m}{24}(m^2-1)(2m^3+3m^2-m-5)$\\
1& $2m\pm 2$ & $ \frac{m}{36}(m^2-1)(2m^3+3m^2-2m-6)$\\
 & $2m\pm 4$ & $ \frac{m}{144}(m-1)(m+1)^2(2m^2+m-9)$\\
 \hline\hline
  & $2m$ & $\frac{m}{480}(m^2-1)(3m^5+6m^4-15m^3-31m^2+12m+33)$\\
 2 & $2m\pm 2$ & $\frac{m}{720}(m^2-1)(m^2-4)(3m^3+6m^2-7m-12)$\\
 & $2m\pm 4$ & $\frac{m}{2880}(m-1)(m+1)^2(m+3)(3m^3-6m^2-16m-31)$\\
 \hline
\end{tabular}\vspace{1.5cm}
\begin{tabular}{|c|c|c|}
\hline
$g$ & $Q$ & $\mathcal N_{\Yboxdim4pt\yng(2),g,Q}^1$\\
\hline\hline
 & $2m\pm 1$ & $\mp \frac{m}{2}(m^3+m^2+3m+1)$\\
& $2m\pm 3$ & $\pm \frac{m}{6}(m-1)(m+1)^2$\\
0 & $3m-2$ & $-\frac{m}{2}(m+1)(2m+1)$\\
 & $3m$ & $ m^2(2m+1)$\\
 & $3m+2$ & $ -\frac{m}{2}(m-1)(2m+1)$\\
 \hline\hline
 & $2m\pm 1$ & $\pm \frac{m}{24}(m+1)(2m^4+m^3+12m^2-m-2)$\\
 & $2m\pm 3$ & $\pm \frac{m}{72}(m-1)(m+2)^2(m+2)(2m-3)$\\
1 & $3m- 2$ & $ \frac{m}{48}(m+1)(2m+1)(9m^2+6m-7)$\\
 & $3m$ & $\frac{m^2}{24}(m+1)(2m+1)(9m-5)$\\
 & $3m+2$ & $ \frac{m}{48}(m^2-1)(2m+1)(9m-7)$\\
 \hline\hline
  & $2m\pm 1$ & $\mp \frac{m}{480}(m^2-1)(3m^5+6m^4+35m^3+48m^2-8m-8)$\\
 & $2m\pm 3$ & $\pm \frac{m}{480}(m-2)(m-1)(m+1)^2(m+2)(m^2+m-4)$\\
2 & $3m- 2$ & $ -\frac{m}{3840}(m^2-1)(226m^4+651m^3-247m^2-259m-149)$\\
 & $3m$ & $\frac{m^2}{5760}(m^2-1)(2m+1)(339m^2+296m-259)$\\
 & $3m+2$ & $ -\frac{m}{11\,520}(m^2-1)(2m+1)(339m^3-215m^2-635m+447)$\\
 \hline
\end{tabular}\vspace{1.5cm}
\begin{tabular}{|c|c|c|}
\hline
$g$ & $Q$ & $\mathcal N_{\Yboxdim4pt\yng(2),g,Q}^2$\\
\hline\hline
 & $2m$ & $-m(2m+1)$\\
0 & $3m\pm 1$ & $\mp m^2(2m+1)$\\
 & $4m$ & $m(2m+1)$\\
 \hline\hline
 & $2m$ & $\pm \frac{1}{6}m(m+1)(2m+1)(2m-1)$\\
1 & $3m\pm 1$ & $\mp \frac{1}{24}m^2(m+1)(9m-5)$\\
 & $4m$ & $ \frac{1}{6}m(m+1)(2m+1)(2m-1)$\\
 \hline\hline
   & $2m$ & $- \frac{m}{90}(m^2-1)(2m+1)(2m-1)(2m+3)$\\
2  & $3m\pm 1$ & $\mp \frac{m^2}{5760}(m^2-1)(2m+1)(339m^3+296m^2-259)$\\
   & $4m$ & $\frac{m}{90}(m^2-1)(2m+1)(2m-1)(2m+3)$\\
 \hline
\end{tabular}\vspace{5mm}
\caption{Integer invariants for the $(2,m)$-torus knot.}\label{tab1}}
\end{table}


\begin{thebibliography}{10}
\providecommand{\url}[1]{\texttt{#1}}
\providecommand{\urlprefix}{URL }
\providecommand{\eprint}[2][]{\url{#2}}

\bibitem{Borhade05a}
P.~Borhade and P.~Ramadevi, 2005.
\newblock \textit{$SO(N)$ reformulated link invariants from topological
  strings}.
\newblock Nucl. Phys. B \textbf{727}, 471--498.
\newblock \href{http://arxiv.org/abs/hep-th/0505008}{{\tt
  arXiv:hep-th/0505008}}.

\bibitem{Bouchard04a}
V.~Bouchard, B.~Florea and M.~Mari{\~n}o, 2004.
\newblock \textit{Counting Higher Genus Curves with Crosscaps in Calabi-Yau
  Orientifolds}.
\newblock J. High Energy Phys. \textbf{12}(35).
\newblock \href{http://arxiv.org/abs/hep-th/0405083}{{\tt
  arXiv:hep-th/0405083}}.

\bibitem{Bouchard05a}
V.~Bouchard, B.~Florea and M.~Mari{\~n}o, 2005.
\newblock \textit{Topological Open String Amplitudes On Orientifolds}.
\newblock J. High Energy Phys. \textbf{2}(2).
\newblock \href{http://arxiv.org/abs/hep-th/0411227}{{\tt
  arXiv:hep-th/0411227}}.

\bibitem{Chen09a}
L.~Chen and Q.~Chen, 2010.
\newblock \textit{Orthogonal Quantum Group Invariants of Links}.
\newblock \href{http://arxiv.org/abs/1007.1656}{{\tt
  arXiv:1007.1656 [math.QA]}}.

\bibitem{Chern74a}
S.-S. Chern and J.~Simons, 1974.
\newblock \textit{Characteristic Forms and Geometric Invariants}.
\newblock Ann. Math. \textbf{99}(1), 48--69.

\bibitem{Fuchs97a}
J.~Fuchs and C.~Schweigert, 1997.
\newblock \textit{Symmetries, Lie Algebras and Representations}.
\newblock Cambridge Monographs on Mathematical Physics (Cambridge University
  Press).

\bibitem{Fulton91a}
W.~Fulton and J.~Harris, 1991.
\newblock \textit{Representation Theory: A First Course} (Springer).

\bibitem{Gopakumar99a}
R.~Gopakumar and C.~Vafa, 1999.
\newblock \textit{On the Gauge Theory/Geometry Correspondence}.
\newblock Adv. Theor. Math. Phys. \textbf{3}, 1415--1443.
\newblock \href{http://arxiv.org/abs/hep-th/9811131}{{\tt
  arXiv:hep-th/9811131}}.

\bibitem{Hadji06a}
R.~J. Hadji and H.~R. Morton, 2006.
\newblock \textit{A basis for the full Homfly skein of the annulus}.
\newblock Math. Proc. Camb. Philos. Soc \textbf{141}, 81--100.
\newblock \href{http://arxiv.org/abs/math/0408078}{{\tt arXiv:math/0408078}}.

\bibitem{Halverson96a}
T.~Halverson, 1996.
\newblock \textit{Characters of the centralizer algebras of mixed tensor
  representations of $GL(r,\mathbb C)$ and the quantum group $\mathcal
  U_q(gl(r,\mathbb C))$.}
\newblock Pacific J. Math. \textbf{174}(2), 359--410.
\newblock \href{http://projecteuclid.org/euclid.pjm/1102365176}{{\tt
  euclid.pjm/1102365176}}.

\bibitem{Isidro93a}
J.~M. Isidro, J.~M.~F. Labastida and A.~V. Ramallo, 1993.
\newblock \textit{Polynomials for Torus Links from Chern--Simons Gauge
  Theories}.
\newblock Nucl. Phys. B \textbf{398}, 187--236.
\newblock \href{http://arxiv.org/abs/hep-th/9210124}{{\tt
  arXiv:hep-th/9210124}}.

\bibitem{Jones85a}
V.~F.~R. Jones, 1985.
\newblock \textit{A Polynomial Invariant for Knots via von Neumann Algebras}.
\newblock Bull. Amer. Math. Soc. \textbf{12}, 103--111.
\newblock \href{http://projecteuclid.org/euclid.bams/1183552338}{{\tt
  euclid.bams/1183552338}}.

\bibitem{Knapp05a}
A.~W. Knapp, 2005.
\newblock \textit{Lie Groups Beyond an Introduction} (Birkh{\"a}user), 2nd
  edition.

\bibitem{Koike89a}
K.~Koike, 1989.
\newblock \textit{On the Decomposition of Tensor Products of the
  Representations of the Classical Groups}.
\newblock Adv. Math. \textbf{74}, 57--86.

\bibitem{Labastida91a}
J.~M.~F. Labastida, P.~M. Llatas and A.~V. Ramallo, 1991.
\newblock \textit{Knot Operators in Chern--Simons Gauge Theory}.
\newblock Nucl. Phys. B \textbf{348}, 651--692.

\bibitem{Labastida95a}
J.~M.~F. Labastida and M.~Mari{\~n}o, 1995.
\newblock \textit{The HOMFLY Polynomial for Torus Links from Chern--Simons Gauge
  Theory}.
\newblock Int. J. Mod. Phys. A \textbf{10}(7), 1045--1089.
\newblock \href{http://arxiv.org/abs/hep-th/9402093}{{\tt
  arXiv:hep-th/9402093}}.

\bibitem{Labastida01a}
J.~M.~F. Labastida and M.~Mari{\~n}o, 2001.
\newblock \textit{Polynomial Invariants for Torus Knots and Topological
  Strings}.
\newblock Comm. Math. Phys. \textbf{217}, 423--449.
\newblock \href{http://arxiv.org/abs/hep-th/0004196}{{\tt
  arXiv:hep-th/0004196}}.

\bibitem{Labastida02a}
J.~M.~F. Labastida and M.~Mari{\~n}o, 2002.
\newblock \textit{A New Point of View in the Theory of Knot and Link
  Invariants}.
\newblock J. Knot Theory Ramif. \textbf{11}, 173--197.
\newblock \href{http://arxiv.org/abs/math/0104180}{{\tt arXiv:math/0104180}}.

\bibitem{Labastida00a}
J.~M.~F. Labastida, M.~Mari{\~n}o and C.~Vafa, 2000.
\newblock \textit{Knots, links and branes at large N}.
\newblock J. High Energy Phys. \textbf{11}(7).
\newblock \href{http://arxiv.org/abs/hep-th/0010102}{{\tt
  arXiv:hep-th/0010102}}.

\bibitem{Labastida96a}
J.~M.~F. Labastida and E.~P{\'e}rez, 1996.
\newblock \textit{A Relation Between the Kauffman and the HOMFLY Polynomials
  for Torus Knots}.
\newblock J. Math. Phys. \textbf{37}(4), 2013--2042.
\newblock \href{http://arxiv.org/abs/q-alg/9507031}{{\tt arXiv:q-alg/9507031}}.

\bibitem{Labastida89a}
J.~M.~F. Labastida and A.~V. Ramallo, 1989.
\newblock \textit{Operator Formalism for Chern--Simons Theories}.
\newblock Phys. Lett. B \textbf{227}(92).

\bibitem{Lin06a}
X.-S. Lin and H.~Zheng, 2006.
\newblock \textit{On the Hecke algebras and the colored HOMFLY polynomial}.
\newblock \href{http://arxiv.org/abs/math/0601267}{{\tt arXiv:math/0601267}}.

\bibitem{Littlewood40a}
D.~E. Littlewood, 1940.
\newblock \textit{The Theory of Group Characters} (Oxford University Press).

\bibitem{Liu07a}
K.~Liu and P.~Peng, 2007.
\newblock \textit{Proof of the Labastida-Mari{\~n}o-Ooguri-Vafa Conjecture}.
\newblock \href{http://arxiv.org/abs/0704.1526}{{\tt arXiv:0704.1526
  [math.QA]}}.

\bibitem{Marino09c}
M.~Mari{\~n}o, 2010.
\newblock \textit{String theory and the Kauffman polynomial}.
\newblock Comm. Math. Phys. \textbf{298}, 613-643.
\newblock \href{http://arxiv.org/abs/0904.1088}{{\tt arXiv:0904.1088
  [hep-th]}}.

\bibitem{Morton08a}
H.~R. Morton and P.~M.~G. Manch\'on, 2008.
\newblock \textit{Geometrical relations and plethysms in the Homfly skein of
  the annulus}.
\newblock J. London Math. Soc. \textbf{78}, 305--328.
\newblock \href{http://arxiv.org/abs/0707.2851}{{\tt arXiv:0707.2851
  [math.GT]}}.

\bibitem{Ooguri00a}
H.~Ooguri and C.~Vafa, 2000.
\newblock \textit{Knot Invariants and Topological Strings}.
\newblock Nucl. Phys. B \textbf{577}, 419--438.
\newblock \href{http://arxiv.org/abs/hep-th/9912123}{{\tt
  arXiv:hep-th/9912123}}.

\bibitem{Paul10a}
C.~Paul, P.~Borhade and P.~Ramadevi, 2010.
\newblock \textit{Composite Invariants and Unoriented Topological String
  Amplitudes}.
\newblock \href{http://arxiv.org/abs/1003.5282}{{\tt arXiv:1003.5282
  [hep-th]}}.\\
 C.~Paul, P.~Borhade and P.~Ramadevi, 2010.
 \newblock\textit{Composite Representation Invariants and Unoriented Topological String Amplitudes}.
 \newblock Nucl. Phys. B \textbf{841}(3), 448-462.

\bibitem{Devi93a}
P.~Ramadevi, T.~Govindarajan and R.~Kaul, 1993.
\newblock \textit{Three Dimensional Chern--Simons Theory as a Theory of Knots
  and Links III : Compact Semi-simple Group}.
\newblock Nucl. Phys. B \textbf{402}, 548--566.
\newblock \href{http://arxiv.org/abs/hep-th/9212110}{{\tt
  arXiv:hep-th/9212110}}.

\bibitem{Sinha00a}
S.~Sinha and C.~Vafa, 2000.
\newblock \textit{SO and Sp Chern--Simons at Large N}.
\newblock \href{http://arxiv.org/abs/hep-th/0012136}{{\tt
  arXiv:hep-th/0012136}}.

\bibitem{Stevan09a}
S.~Stevan, 2009.
\newblock \textit{Knot Operators in Chern--Simons Gauge Theory}.
\newblock Master's thesis, University of Geneva.

\bibitem{Witten89a}
E.~Witten, 1989.
\newblock \textit{Quantum Field Theory and the Jones Polynomial}.
\newblock Comm. Math. Phys. \textbf{121}, 351--399.
\newblock \href{http://projecteuclid.org/euclid.cmp/1104178138}{{\tt
  euclid.cmp/1104178138}}.

\end{thebibliography}
\end{document}